\documentclass[aps,pra,twocolumn,amssymb,amsmath,amsfonts,nofootinbib,floatfix,superscriptaddress]{revtex4-2}

\usepackage[english]{babel}
\usepackage{amsthm}
\usepackage[toc,page]{appendix}
\usepackage[colorlinks=true,citecolor=blue,linkcolor=magenta]{hyperref}
\usepackage{bbm}
\usepackage{graphicx}
\usepackage{stackengine,xcolor}
\usepackage{nicematrix}
\usepackage{cleveref}
\usepackage{mathtools}
\usepackage{quantikz}
\usepackage{adjustbox}
\newcommand{\poly}{\operatorname{poly}}

\newcommand{\Tr}{{\rm Tr}}

\newcommand{\spf}{\mathfrak{sp}(M,\mathbb{R})}

\newcommand{\s}{\vec{\sigma}}
\newcommand{\z}{\hat{\vec{z}}}
\newcommand{\ra}{\rangle}
\newcommand{\la}{\langle}

\renewcommand{\vec}[1]{\boldsymbol{#1}}

\newcommand{\id}{\openone}
\newcommand{\ad}{^{\dagger}}
\newcommand{\ketbra}[2]{|#1\rangle\!\langle #2|}

\newcommand{\FC}{\mathcal{F}}

\newcommand{\HC}{\mathcal{H}}

\newcommand{\OC}{\mathcal{O}}

\newtheorem{theorem}{Theorem}
\newtheorem{problem}{Problem}
\newtheorem{lemma}{Lemma}

\newtheorem{definition}{Definition}
\newtheorem{supproposition}{Supplemental Proposition}

\newcommand{\mbb}[1]{\mathbb{#1}}

\def\opq{\hat{q}}
\def\opp{\hat{p}}

\def\opa{\hat{a}}
\def\opb{\hat{a}^{\dagger}}

\begin{document}

\title{Gate-based quantum simulation of Gaussian bosonic circuits \\ on exponentially many modes}

\author{Alice Barthe}
\affiliation{CERN, Meyrin, Geneva 1211, Switzerland}
\affiliation{Theoretical Division, Los Alamos National Laboratory, Los Alamos, New Mexico 87545, USA}
\affiliation{Instituut-Lorentz, Universiteit Leiden, Leiden, 2333CA, the Netherlands}

\author{M. Cerezo}
\thanks{cerezo@lanl.gov}
\affiliation{Information Sciences, Los Alamos National Laboratory, Los Alamos, New Mexico 87545, USA}
\affiliation{Quantum Science Center, Oak Ridge, TN 37931, USA}

\author{Andrew T. Sornborger}
\affiliation{Information Sciences, Los Alamos National Laboratory, Los Alamos, New Mexico 87545, USA}

\author{Mart\'in Larocca}
\affiliation{Theoretical Division, Los Alamos National Laboratory, Los Alamos, New Mexico 87545, USA}
\affiliation{Center for Non-Linear Studies, Los Alamos National Laboratory, Los Alamos, New Mexico 87545, USA}

\author{Diego Garc\'ia-Mart\'in}
\affiliation{Information Sciences, Los Alamos National Laboratory, Los Alamos, New Mexico 87545, USA}

\begin{abstract}
We introduce a framework for simulating, on an $(n+1)$-qubit quantum computer, the action of a Gaussian Bosonic (GB) circuit on a state over $2^n$ modes. Specifically, we encode the initial bosonic state's expectation values over quadrature operators (and their covariance matrix) as an input qubit-state. This is then evolved by a quantum circuit that effectively implements the symplectic propagators induced by the GB gates. We find families of GB circuits and initial states leading to efficient quantum simulations. For this purpose, we introduce a dictionary that maps between GB and qubit gates such that particle- (non-particle-) preserving  GB gates lead to real (imaginary) time evolutions at the qubit level. For the special case of particle-preserving circuits, we present a BQP-complete GB decision problem, indicating that GB evolutions of Gaussian states on exponentially many modes are as powerful as universal quantum computers. We also perform numerical simulations of an interferometer on $\sim8$ billion modes, illustrating the power of our framework. 
\end{abstract}

\maketitle

\textbf{Introduction.} The study of the power of quantum computers has been a central and fundamental topic in complexity theory. While it is known that these devices can solve problems at least as fast as classical probabilistic computers, we also know that they are unable to do so more than exponentially faster~\cite{bernstein1997quantum,aaronson2010bqp}. This has raised the question: \textit{What are the tasks that saturate this separation? I.e., what are the problems for which a quantum computer can achieve an exponential advantage over its classical counterparts?}      

To answer such questions one starts with the class Bounded-Error Quantum Polynomial (BQP),  the set of decision problems that a quantum computer can solve in polynomial time with a small constant probability of failure. Then, one determines the subset of problems that are the hardest therein, known as BQP-complete. Several BQP-complete problems are known, such as those based on the Harrow–Hassidim–Lloyd algorithm~\cite{harrow2009quantum},  scattering in scalar quantum field theory~\cite{jordan2018bqp}, and more recently on the quantum simulation of exponentially many coupled classical oscillators~\cite{babbush2023exponential}. The latter presents the intriguing perspective that simulating exponentially large linear and energy-preserving simple classical systems leads to BQP-completeness, and, under reasonable complexity theory assumptions, to an exponential quantum advantage over classical methods. 

In this work we prove an analogous result to that in Ref.~\cite{babbush2023exponential}, namely, that the quantum simulation of particle-preserving GB circuits (i.e. passive linear optics)~\cite{knill2001scheme,bouland2014generation,sawicki2015universality} acting on Gaussian initial states on exponentially many modes also leads to BQP-completeness (see Fig.~\ref{fig:schematic}). At its core, our framework starts with the realization that a direct simulation of a bosonic system is intractable on a qubit-based quantum computer, as the associated Hilbert spaces are fundamentally different (with one being infinite-dimensional, and the other discrete). This issue can be avoided by restricting the simulation to the first and second moments of the quadrature operators. That is, we encode in a quantum state the position and momentum expectation values (and their covariance matrix) over the initial bosonic state. As such, instead of simulating the action of the GB circuit on the bosonic Hilbert space, we implement its effective action on the expectation values on a gate-based quantum computer. 

In this context, we present a constructive dictionary that translates back-and-forth between the symplectic propagator associated with a universal set of GB gates (beamsplitters, phase and squeezing gates) and qubit circuits. The efficiency of our simulation framework relies on several key conditions, such as the input qubit-state being preparable in polynomial time, and the quantum circuit requiring only polynomially-many gates. Indeed, we present cases of interest for which these two conditions are satisfied, and therefore for which we can achieve an exponential quantum advantage.

\textbf{Background}. In what follows, we will consider systems composed of $M$ bosonic modes (with $M=2^n$). Let  $\hat{a}^\dagger_m$ and $\hat{a}_m$, with $m=1,\ldots,M$, respectively denote the creation and annihilation operators for the $m$-th mode~\cite{braunstein2005quantum}. We consider the standard Hermitian \textit{quadrature operators}, position $\opq_m=\frac{1}{\sqrt{2}}(\hat{a}_m+\hat{a}_m\ad)$ and momentum $\opp_m=\frac{i}{\sqrt{2}}(\hat{a}_m\ad-\hat{a}_m)$. They satisfy the canonical commutation relations $[\opq_m,\opp_{m'}]=i\delta_{mm'}$. Furthermore, we will focus on the case where an $M$-mode bosonic state $\rho_0$  evolves under the action of a GB circuit whose gates are generated by time-independent Hamiltonians that are quadratic in the position and momentum operators\footnote{A generalization to time-dependent Hamiltonians is direct using standard Hamiltonian-simulation techniques~\cite{poulin2011quantum}.}. 

These GB generators, also known as free-bosonic generators, are arbitrary real-valued degree-two homogeneous polynomials on the quadrature operators 
\begin{equation}
    \hat{H} = \frac{1}{2} \z^T K \z\,, \text{ with } \z=(\opq_1,\dots, \opq_M, \opp_1,\dots, \opp_M)^T\,,
\end{equation} 
where $K$ is a real $2M\times 2M$ symmetric matrix. The vector $\z$ allows us to express the commutation relations in the compact form $[\z_\alpha,\z_\beta]=i\Omega_{\alpha\beta}$, where $\Omega= iY\otimes \id_M$. Here, $Y$ is the usual $2\times 2$ Pauli matrix and $\id_M$ the $M\times M$ identity matrix.

\begin{figure}[t]
    \centering
    \includegraphics[width=1\columnwidth]{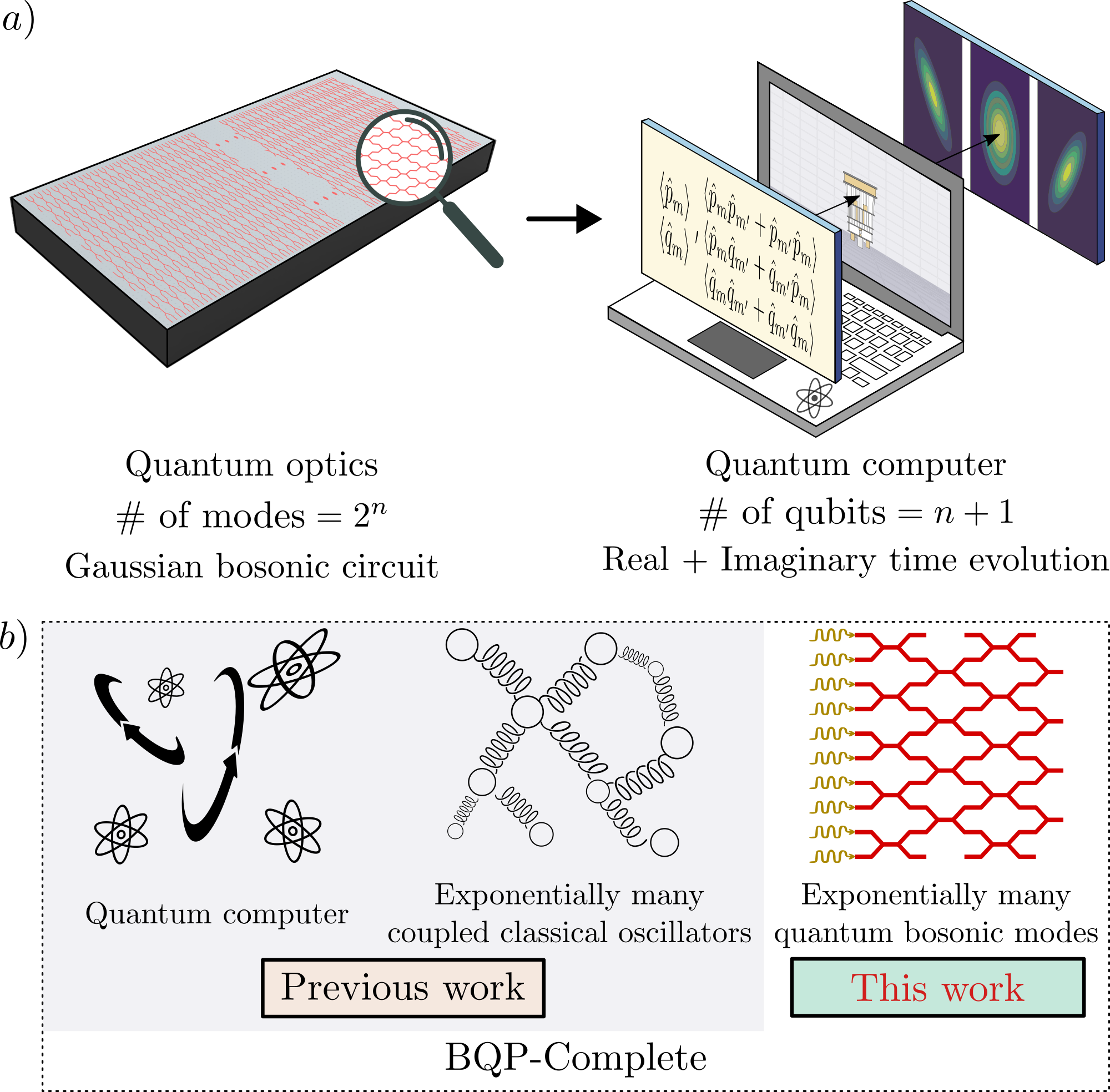}
    \caption{\textbf{Schematic representation of our main results.} a) We present a framework for simulating the action of a GB circuit on the first and second moments of quadrature operators of a bosonic state on $2^n$ modes on an $(n+1)$-qubit gate-based quantum computer. b) We show that particle-preserving GB evolutions on Gaussian bosonic states are sufficient to define a problem that is BQP-complete, thus indicating that passive linear optics on exponentially many bosonic modes are as powerful as universal quantum computers.   }
    \label{fig:schematic}
\end{figure}

Next, we collect the expectation values of the position and momentum operators over $\rho_0$ in a vector $\la \z\ra =(\la \opq_1\ra,\dots,\la \opq_M\ra,\la \opp_1\ra,\dots,\la \opp_M\ra)^T \in \mbb{R}^{2M}$, with $\langle \hat{x}\rangle=\Tr[\rho_0 \hat{x}]$.   As shown in the  Supplemental Information (SI),  evolving  $\rho_0$ with a unitary generated by a GB generator $\hat{H}$ for a time $t$ induces the evolution of $\la \z\ra$ as
\begin{equation}
     \frac{\partial \langle \z \rangle}{\partial t} =  \Omega K \la\z\ra \,,\,\text{ so that }\, \la\z\ra(t) = e^{t \Omega K} \la\z\ra(0)\,.
    \label{eq:exp-val-motion}
\end{equation} 
Here,  $\la\z\ra(t)$ denotes the vector containing the expectation values of positions and momenta at time $t$, and thus $\la\z\ra(0)$ represents the initial condition. Since the canonical commutation relations must be preserved, the propagator $e^{t \Omega K}$ is a $2M\times 2M$ symplectic matrix with real entries belonging to the Lie group $\mathbb{SP}(M,\mathbb{R})$~\cite{braunstein2005quantum}, which in turn implies that $\Omega K$ is an operator in the symplectic Lie algebra $\mathfrak{sp}(M,\mathbb{R})$.

Note that while, in general, $e^{t \Omega K}$ is not unitary, we characterize the quadratic Hamiltonians leading to unitary evolutions of the vector $\la \z\ra$, when a gate generator is of the form $\hat{H} = \sum_{m,m'=1}^{M} h_{m m'} \,\hat{a}_m^\dagger \hat{a}_{m'} +\frac{\Tr[h]}{2}\id_{2M}$, 
with $h$ a Hermitian matrix. Then $[\Omega, K]=0$, and the propagator $e^{t\Omega K}$ is the real-time evolution of the Hermitian $i \Omega K$. Hamiltonians of this form are known as \textit{particle preserving} since the hopping terms $\hat{a}_m^\dagger \hat{a}_{m'}$  move a boson from mode $m'$ to mode $m$. 
 We also characterize \textit{non-particle-preserving Hamiltonians} of the form $\hat{H} = \sum_{m,m'=1}^{M} \Delta_{m m'}^\dagger \,\hat{a}_m \hat{a}_{m'} +\text{h.c.}$, whose propagator corresponds to the imaginary-time evolution of  $\la\z\ra$ under the effective Hamiltonian $-\Omega K$. 

In addition to $\la\z\ra$, we also collect the expectation value of products of quadrature operators over $\rho_0$ in the $2M\times 2M$ positive-definite covariance matrix $\s$ whose entries are given by $\s_{\alpha \beta}=\frac{1}{2}\langle \z_\alpha \z_\beta+ \z_\beta\z_\alpha\rangle -\langle \z_\alpha\rangle\langle \z_\beta\rangle$~\cite{adesso2014continuous}. 
Analogously to Eq.~\eqref{eq:exp-val-motion}, we find  (see the SI) that  
\begin{equation}
   \frac{\partial \s}{ \partial t} = \Omega K \s-\s K\Omega\,,\,\text{ so that }\, \s(t) = e^{\Omega K t}\s(0)  e^{\mp\Omega K t}\,, \nonumber
\end{equation}
where the $- \,(+)$ sign corresponds to particle (non-particle) preserving Hamiltonians, as defined above. 

The previous insights pave the way to simulate the action of GB circuits on a gate-based quantum computer.

\begin{figure*}[t]
    \centering
\includegraphics[width=1\linewidth]{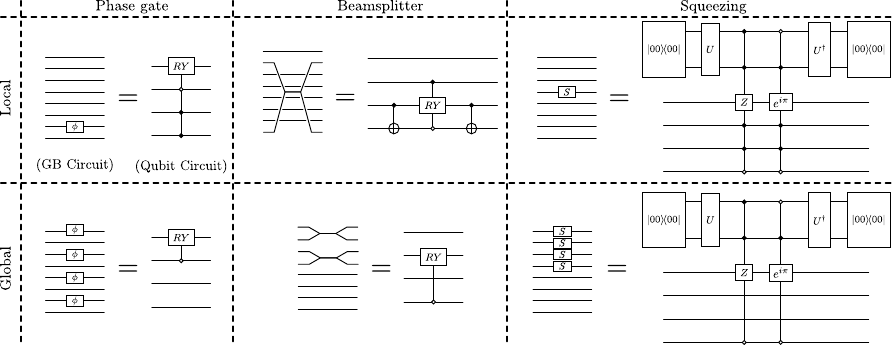}
    \caption{\textbf{Examples of GB gates in the qubit picture.} We consider a bosonic system on $M=8$ modes, leading to a circuit on $4$ qubits. The local phase gate acts on the mode $m=6=2^0\times0+2^1\times1+2^2\times1$.  The local beamsplitter acts on the modes $m=1$ and $m'=7$, whose binary representations only share the least significant bit. The local squeezing gate acts on mode $m=3=2^0\times1+2^1\times1+2^2\times0$, and is represented by an imaginary time evolution as a linear combination of unitaries with post-selection on two ancillary qubits (which have been added on top). The global phase gate acts on all modes whose index is even. 
    The global beamsplitter is applied to the first half of the modes, pairing each mode with even index $m$ to its nearest neighbor mode with index $m'=m+1$.  The global squeezing gate is applied to the first half of the modes.}
\label{fig:gates}
\end{figure*}

\textbf{Initialization}.
To start the simulation on the quantum computer, we encode the normalized vector $\la \z\ra$ (normalized matrix $\s$) into a pure (mixed) $(n+1)$-qubit quantum state $\ket{\z}\propto \la \z\ra$ ($\varrho_{\s}\propto \s$). For instance, we have  
\begin{equation}
    \ket{\z} = \frac{1}{{\lVert\la \z\ra\rVert}_2}\sum_{m=1}^{2^n} \la \hat{q}_m\ra \ket{0}\otimes\ket{m} + \la \hat{p}_m\ra \ket{1}\otimes\ket{m}\,.\label{eq:encoding}
\end{equation}
For our scheme to be efficient, such states need to be preparable in polynomial time. Such cases occur, e.g., when there are $\OC(\poly(n))$ non-zero known position and momentum expectation values. Equation~\eqref{eq:encoding} illustrates the privileged role of the first qubit~\cite{garcia2024architectures}, separating the Hilbert space into two subspaces, one associated with the positions and the other with the momenta. We will henceforth refer to the first qubit as the \textit{symplectic qubit}. The remaining $n$ qubits serve as a register for each of the $2^n$ modes, and we will refer to them as \textit{register qubits}. In particular, in Eq.~\eqref{eq:encoding} each mode label $m$ is encoded via its binary decomposition as $m=2^0 m_1 + 2^1 m_2+\ldots+2^{n-1} m_{n}$, with $m_l\in\{0,1\}$, and with the least significant qubit being the top-most register qubit. 

Here we note that when $\rho_0$  is Gaussian, then it is fully characterized by the first and second moments of the quadrature operators, meaning that $\ket{\z}$ and $\varrho_{\s}$ provide a full description of these states. Moreover, if the initial state is also coherent (eigenstates of $\opa$), then $\varrho_{\s}$ is the maximally mixed state on $(n+1)$ qubits. For other non-Gaussian states our framework is restricted to the information contained in $\z$ and $\s$.

\textbf{Evolution}.
Once the initial state is prepared, we can evolve it with a quantum circuit that effectively implements the symplectic propagators associated with the gates in the GB circuit (see for instance Eq.~\eqref{eq:exp-val-motion}). In particular, we present a dictionary to efficiently translate between standard GB and qubit gates. We also refer the reader to Fig.~\ref{fig:gates} for an explicit circuit depiction of some GB gates in qubit circuit form.

The \textit{phase gate} is a particle-preserving gate acting on a single mode $m$. In terms of bosonic operators, its generator is $\hat{H}=\opp_m^2 + \opq_m^2$, yielding $\Omega K= 2iY \otimes \ketbra{m}{m}$ in the qubit picture. This corresponds to an $R_y$  gate (rotation about the $y$-axis) on the symplectic qubit, conditioned on the $\ketbra{m}{m}$ state on the register.

The \textit{beamsplitter} is a particle-preserving gate acting on two modes $m$ and $m'$. It is generated by $\hat{H}=\opq_m\opp_{m'} - \opq_{m'}\opp_{m}$  (with $m\neq m'$), which results in $\Omega K= 2i\id \otimes (i\ketbra{m}{m'} - i\ketbra{m'}{m})$. This corresponds to an $R_y$ rotation in the subspace spanned by $\ket{m}$ and $\ket{m'}$. This gate can be implemented by using controlled-not gates to transform $\ket{m}$ and $\ket{m'}$ into two computational-basis states that only differ in one qubit, then performing a controlled-$R_y$ rotation on that qubit (conditioned on the other $n-1$ qubits in the register), and applying the same controlled-not gates to return to the original basis. We highlight the fact that this gate acts trivially on the symplectic qubit, as beamsplitters conserve total momentum and total position.

The \textit{squeezing gate} is a non-particle-preserving gate that acts on a single mode $m$. Its generator is $\hat{H}=\pm \left(\opq_m\opp_m + \opp_m\opq_m\right)$, leading to $\Omega K= \pm  2 X \otimes \ketbra{m}{m}$. This produces an imaginary-time evolution~ \cite{motta2020determining,zoufal2021variational,zoufal2021error,silva2023fragmented}, which can be implemented, e.g., with two ancillary qubits and post-selection. As illustrated in Fig.~\ref{fig:gates} and further discussed in the SI, we propose to implement it as a linear combination of unitaries~\cite{childs2012hamiltonian}, with the latter being multi- controlled-Z and controlled-phase gates.  Notably, the fact that states cannot be arbitrarily squeezed translates in our framework as the success probability of the imaginary-time evolution for infinite time being zero.  

Finally, the \textit{displacement gate}, a common non-particle-preserving Gaussian gate, cannot be included as a linear qubit gate in the proposed framework (see the SI for an additional discussion on this gate).

We stress that while our framework requires the implementation of multi-controlled qubit operations to simulate the action of some GB gates, these can be compiled exactly using only $\OC(n)$ local gates~\cite{nielsen2000quantum}. Moreover, if instead of one or two modes, we consider GB gates that act on several (potentially exponentially many) modes, there can be simplifications that render them much easier to implement at the qubit level. As an important example, we show in Fig.~\ref{fig:gates} (see also SI) how some global phase gates and beamsplitters acting on $2^{n-1}$ modes simplify to two-qubit controlled-$R_y$ rotations, as well as how the number of control qubits is reduced for some global squeezing gates. This means that when the GB circuit is composed of polynomially many such local particle-preserving gates, the implementation of the quantum circuit is efficient. When non-particle-preserving squeezing gates are added, then the efficiency will ultimately depend on how many such gates are added, as well as on their squeezing strength parameters.

At this point, we recall that the covariance matrix of coherent states leads to a maximally mixed state $\varrho_{\s}$, and therefore it remains invariant when evolved with a purely unitary circuit (e.g., particle-preserving GB gates in an interferometer). This result is well aligned with the bosonic picture where coherent states going through interferometers remain coherent states. Then, as squeezing gates are not particle-preserving, their action on a coherent state corresponds to purifying the associated $\varrho_{\s}$.

\textbf{Measurements}. At the output of the quantum circuit, we obtain states that represent the evolved expectation values and covariance matrix of quadrature operators (which we respectively denote as $\ket{\z'}$ and $\varrho_{\s'}$). We now discuss how measuring these states allows us to extract useful information about the GB circuit. 

There are a variety of measurements of interest for Gaussian states. As detailed in the SI,  photon counting can be implemented for coherent states by sampling bitstrings from  $\ket{\z'}$. To understand why this is the case, we recall that the energy for each mode is proportional to the probability of sampling the corresponding bitstring.
Secondly, for any bosonic state homodyne measurements correspond to retrieving the position (momentum) of a specific mode. At the qubit level, these measurements can be implemented from a swap-Hadamard test between $\ket{\z'}$ and some computational-basis state of interest. Then, our framework also allows us to estimate the fraction of total momentum and total position, as this value can be retrieved by measuring the symplectic qubit in $\ket{\z'}$. Similarly, the fractional energy of the first and second half of the modes can be estimated by measuring the bottom-most (most significant) register qubit. Moreover, we note that combining computational basis measurements (or Hadamard tests) on $\ket{\z'}$ and $\varrho_{\s'}$ allows us to estimate the energy in a given mode. To finish, we note that while the total energy remains constant for particle-preserving GB circuits, this quantity can change when non-particle-preserving gates are included. In this case, while we cannot directly measure the total energy of the system (as the states need to be normalized), one can keep track of the total energy by using as a proxy the success probability of the imaginary-time evolutions.

\textbf{BQP-completeness}. 
 We now show that we can leverage our framework to devise a decision problem based on a restricted class of large optical quantum interferometers and prove that it is BQP-complete. We begin by introducing a family of quantum interferometers, that we refer to as  \textit{bit-structured}.

\begin{definition}[Bit-structured interferometer] A bit-structured interferometer acting on $2^n$ nodes consists of $L$ global beamsplitters, such that each global beamsplitter acts on $2^{n-1}$ modes. A global beamsplitter is specified by two natural numbers, $k\neq l$, between 1 and $n$. The global beamsplitter then acts on all the modes with indices $\{m\}$ such that their $k$-th bit is equal to 0, by applying local beamsplitters between modes with indices $m,m'$ that only differ in their $l$-th bit.
\label{def:interferometer}
\end{definition}

Our decision problem is then phrased as follows.

\begin{problem}
\label{pbm:interferometer} 
    Consider a bit-structured interferometer (see \Cref{def:interferometer}) acting on $2^n$ modes with $L\in\OC(\poly(n))$, 
    and an input state such that the first mode is displaced in position by a real constant $x$ while the state of the remaining modes is the vacuum. Then, decide whether the expectation value of the position on the first mode at the output of the interferometer is
    \begin{align}
        1.\,\langle\opq_1\rangle > \frac{2}{3} x\,,\qquad {\rm or} \qquad 2.\, \langle\opq_1\rangle < \frac{1}{3} x\,,\nonumber
    \end{align}
    given the promise that either one or the other is true. 
\end{problem}

\Cref{pbm:interferometer} is illustrated in \Cref{fig:bqp}, where we simulate an interferometer with $\sim8$ billion modes (i.e., a $33$-qubit circuit). There, we keep track of $\la \opq_1\ra / x$ (which corresponds to the overlap with the $\ket{0}^{\otimes n}$ state) as the state evolves through the beamsplitters. 

Our main result is the next theorem.
\begin{theorem}\label{theo:1}
     \Cref{pbm:interferometer} is BQP-complete.
\end{theorem}

\begin{figure}[t!]
    \centering
    \includegraphics[width=1\columnwidth]{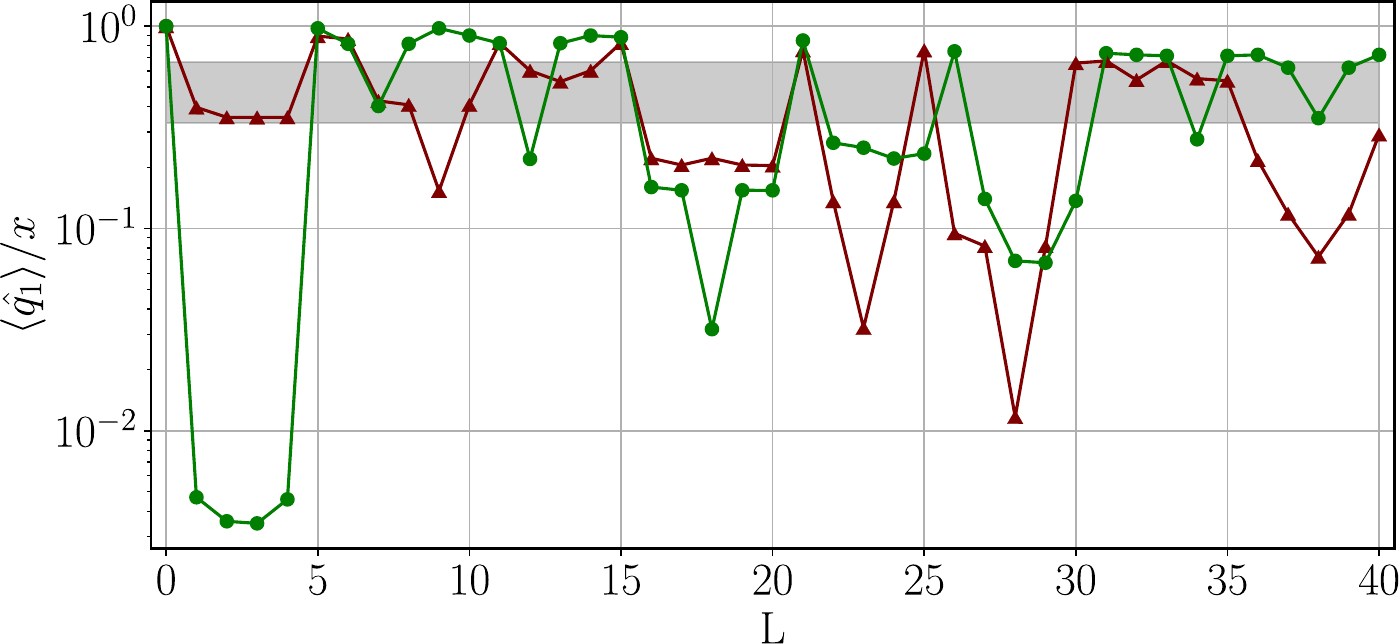}
    \caption{ \textbf{Simulation of a structured interferometer on $\sim8$ billion modes}. We illustrate \Cref{pbm:interferometer} by tracking two non-trivial evolutions of $\la \opq_1\ra / x$ along a large bit-structured interferometer, the green (red) plot corresponding to a YES (NO) instance. The gray region corresponds to $\frac{1}{3}<\la \opq_1\ra / x < \frac{2}{3}$. The simulations were performed with  \texttt{Qibo}~\cite{efthymiou2020qibo,efthymiou2022quantum}. }
    \label{fig:bqp}
\end{figure}

The proof of \Cref{theo:1} can be found in the SI, but we give here a summary of the key points. First, we show that \Cref{pbm:interferometer} is contained in BQP. To do this, we simply use the mapping between beamsplitters and qubit gates explained previously to prove that a bit-structured interferometer with $\OC(\poly(n))$ layers always results in a polynomial-size quantum circuit, which is a necessary condition to be contained in BQP.

Finally, we need to prove that \Cref{pbm:interferometer} is BQP-Hard. That is, if we can solve it, we can also solve any other problem in BQP with an additional overhead that is polynomial in $n$. This is done by first showing that each global beamsplitter gives rise to a controlled-$R_y$ gate (acting when the control qubit is in the $\ket{0}$ state). Then, we use the result from~\cite{rudolph20022} which states that  controlled-$R_y$ rotations constitute a universal gate set, in the sense that any quantum computation can be performed with these gates~\cite{shi2002both,aharonov2003simple}.

\textbf{Outlook}. This work contributed to the body of knowledge of schemes leading to BQP-complete tasks. In particular, along with Ref.~\cite{babbush2023exponential}, we show examples of how linear and energy-preserving evolutions of exponentially many simple physical systems can be efficiently simulated on a quantum computer. A key unique feature of our framework is that we translate a sequence of gates from one physical system to another (product of exponentials) rather than performing real-time evolution of a Hamiltonian (exponential of a sum).

Importantly, our work opens up several interesting research directions. For instance, we have shown how to simulate the evolution of the first and second moments of bosonic states. This makes the simulation complete only for Gaussian states. Hence, we envision two possible paths to extend our proposed framework to better approximate the simulation of non-Gaussian states. First, we could find ways to simulate the evolution of higher moments. To what extent this would improve the quality of a simulation for non-Gaussian states, such as Fock states, remains open. Second, because bosonic states can be written as a continuous sum over coherent states, approximating them on a grid of coherent states may also be a viable strategy to simulate non-Gaussian states.  Finally, we also leave for future work a precise characterization of the computational complexity of the quantum simulation of non-particle preserving GB circuits, which are likely outside of BQP.

\section*{Acknowledgments}

We thank Samuel Slezak, Christa Zoufal, Paolo Braccia and Nathan Killoran for insightful conversations. A.B. was initially supported by the U.S. Department of Energy through a quantum computing program sponsored by the Los Alamos National Laboratory (LANL) Information Science \& Technology Institute. A.B. also acknowledges support by CERN through the CERN Quantum Technology Initiative and by the quantum computing for earth observation (QC4EO) initiative of ESA $\phi$-lab, partially funded under contract 4000135723/21/I-DT-lr, in the FutureEO program.  M.C. and A.S. were initially supported by LANL's ASC Beyond Moore’s Law project.  M.L.
was supported by the Center for Nonlinear Studies at
LANL and by the Laboratory Directed Research and
Development (LDRD) program of LANL under project
number 20230049DR. D.G.M. was supported by LDRD program of LANL under project number 20230527ECR. 
This work was also supported by the Quantum Science Center (QSC), a National Quantum Information Science Research Center of the U.S. Department of Energy (DOE).

\bibliography{quantum.bib}

\begin{thebibliography}{24}%
\makeatletter
\providecommand \@ifxundefined [1]{%
 \@ifx{#1\undefined}
}%
\providecommand \@ifnum [1]{%
 \ifnum #1\expandafter \@firstoftwo
 \else \expandafter \@secondoftwo
 \fi
}%
\providecommand \@ifx [1]{%
 \ifx #1\expandafter \@firstoftwo
 \else \expandafter \@secondoftwo
 \fi
}%
\providecommand \natexlab [1]{#1}%
\providecommand \enquote  [1]{``#1''}%
\providecommand \bibnamefont  [1]{#1}%
\providecommand \bibfnamefont [1]{#1}%
\providecommand \citenamefont [1]{#1}%
\providecommand \href@noop [0]{\@secondoftwo}%
\providecommand \href [0]{\begingroup \@sanitize@url \@href}%
\providecommand \@href[1]{\@@startlink{#1}\@@href}%
\providecommand \@@href[1]{\endgroup#1\@@endlink}%
\providecommand \@sanitize@url [0]{\catcode `\\12\catcode `\$12\catcode `\&12\catcode `\#12\catcode `\^12\catcode `\_12\catcode `\%12\relax}%
\providecommand \@@startlink[1]{}%
\providecommand \@@endlink[0]{}%
\providecommand \url  [0]{\begingroup\@sanitize@url \@url }%
\providecommand \@url [1]{\endgroup\@href {#1}{\urlprefix }}%
\providecommand \urlprefix  [0]{URL }%
\providecommand \Eprint [0]{\href }%
\providecommand \doibase [0]{https://doi.org/}%
\providecommand \selectlanguage [0]{\@gobble}%
\providecommand \bibinfo  [0]{\@secondoftwo}%
\providecommand \bibfield  [0]{\@secondoftwo}%
\providecommand \translation [1]{[#1]}%
\providecommand \BibitemOpen [0]{}%
\providecommand \bibitemStop [0]{}%
\providecommand \bibitemNoStop [0]{.\EOS\space}%
\providecommand \EOS [0]{\spacefactor3000\relax}%
\providecommand \BibitemShut  [1]{\csname bibitem#1\endcsname}%
\let\auto@bib@innerbib\@empty
\bibitem [{\citenamefont {Bernstein}\ and\ \citenamefont {Vazirani}(1997)}]{bernstein1997quantum}%
  \BibitemOpen
  \bibfield  {author} {\bibinfo {author} {\bibfnamefont {E.}~\bibnamefont {Bernstein}}\ and\ \bibinfo {author} {\bibfnamefont {U.}~\bibnamefont {Vazirani}},\ }\bibfield  {title} {\bibinfo {title} {Quantum complexity theory},\ }\href {https://doi.org/https://doi.org/10.1137/S0097539796300921} {\bibfield  {journal} {\bibinfo  {journal} {SIAM Journal on computing}\ }\textbf {\bibinfo {volume} {26}},\ \bibinfo {pages} {1411} (\bibinfo {year} {1997})}\BibitemShut {NoStop}%
\bibitem [{\citenamefont {Aaronson}(2010)}]{aaronson2010bqp}%
  \BibitemOpen
  \bibfield  {author} {\bibinfo {author} {\bibfnamefont {S.}~\bibnamefont {Aaronson}},\ }\bibfield  {title} {\bibinfo {title} {Bqp and the polynomial hierarchy},\ }in\ \href {https://doi.org/10.1145/1806689.1806711} {\emph {\bibinfo {booktitle} {Proceedings of the forty-second ACM symposium on Theory of computing}}}\ (\bibinfo {year} {2010})\ pp.\ \bibinfo {pages} {141--150}\BibitemShut {NoStop}%
\bibitem [{\citenamefont {Harrow}\ \emph {et~al.}(2009)\citenamefont {Harrow}, \citenamefont {Hassidim},\ and\ \citenamefont {Lloyd}}]{harrow2009quantum}%
  \BibitemOpen
  \bibfield  {author} {\bibinfo {author} {\bibfnamefont {A.~W.}\ \bibnamefont {Harrow}}, \bibinfo {author} {\bibfnamefont {A.}~\bibnamefont {Hassidim}},\ and\ \bibinfo {author} {\bibfnamefont {S.}~\bibnamefont {Lloyd}},\ }\bibfield  {title} {\bibinfo {title} {Quantum algorithm for linear systems of equations},\ }\href {https://doi.org/10.1103/PhysRevLett.103.150502} {\bibfield  {journal} {\bibinfo  {journal} {Physical Review Letters}\ }\textbf {\bibinfo {volume} {103}},\ \bibinfo {pages} {150502} (\bibinfo {year} {2009})}\BibitemShut {NoStop}%
\bibitem [{\citenamefont {Jordan}\ \emph {et~al.}(2018)\citenamefont {Jordan}, \citenamefont {Krovi}, \citenamefont {Lee},\ and\ \citenamefont {Preskill}}]{jordan2018bqp}%
  \BibitemOpen
  \bibfield  {author} {\bibinfo {author} {\bibfnamefont {S.~P.}\ \bibnamefont {Jordan}}, \bibinfo {author} {\bibfnamefont {H.}~\bibnamefont {Krovi}}, \bibinfo {author} {\bibfnamefont {K.~S.}\ \bibnamefont {Lee}},\ and\ \bibinfo {author} {\bibfnamefont {J.}~\bibnamefont {Preskill}},\ }\bibfield  {title} {\bibinfo {title} {Bqp-completeness of scattering in scalar quantum field theory},\ }\href {https://doi.org/10.22331/q-2018-01-08-44} {\bibfield  {journal} {\bibinfo  {journal} {Quantum}\ }\textbf {\bibinfo {volume} {2}},\ \bibinfo {pages} {44} (\bibinfo {year} {2018})}\BibitemShut {NoStop}%
\bibitem [{\citenamefont {Babbush}\ \emph {et~al.}(2023)\citenamefont {Babbush}, \citenamefont {Berry}, \citenamefont {Kothari}, \citenamefont {Somma},\ and\ \citenamefont {Wiebe}}]{babbush2023exponential}%
  \BibitemOpen
  \bibfield  {author} {\bibinfo {author} {\bibfnamefont {R.}~\bibnamefont {Babbush}}, \bibinfo {author} {\bibfnamefont {D.~W.}\ \bibnamefont {Berry}}, \bibinfo {author} {\bibfnamefont {R.}~\bibnamefont {Kothari}}, \bibinfo {author} {\bibfnamefont {R.~D.}\ \bibnamefont {Somma}},\ and\ \bibinfo {author} {\bibfnamefont {N.}~\bibnamefont {Wiebe}},\ }\bibfield  {title} {\bibinfo {title} {Exponential quantum speedup in simulating coupled classical oscillators},\ }\href {https://doi.org/10.1103/PhysRevX.13.041041} {\bibfield  {journal} {\bibinfo  {journal} {Physical Review X}\ }\textbf {\bibinfo {volume} {13}},\ \bibinfo {pages} {041041} (\bibinfo {year} {2023})}\BibitemShut {NoStop}%
\bibitem [{\citenamefont {Knill}\ \emph {et~al.}(2001)\citenamefont {Knill}, \citenamefont {Laflamme},\ and\ \citenamefont {Milburn}}]{knill2001scheme}%
  \BibitemOpen
  \bibfield  {author} {\bibinfo {author} {\bibfnamefont {E.}~\bibnamefont {Knill}}, \bibinfo {author} {\bibfnamefont {R.}~\bibnamefont {Laflamme}},\ and\ \bibinfo {author} {\bibfnamefont {G.~J.}\ \bibnamefont {Milburn}},\ }\bibfield  {title} {\bibinfo {title} {A scheme for efficient quantum computation with linear optics},\ }\href {https://doi.org/10.1038/35051009} {\bibfield  {journal} {\bibinfo  {journal} {nature}\ }\textbf {\bibinfo {volume} {409}},\ \bibinfo {pages} {46} (\bibinfo {year} {2001})}\BibitemShut {NoStop}%
\bibitem [{\citenamefont {Bouland}\ and\ \citenamefont {Aaronson}(2014)}]{bouland2014generation}%
  \BibitemOpen
  \bibfield  {author} {\bibinfo {author} {\bibfnamefont {A.}~\bibnamefont {Bouland}}\ and\ \bibinfo {author} {\bibfnamefont {S.}~\bibnamefont {Aaronson}},\ }\bibfield  {title} {\bibinfo {title} {Generation of universal linear optics by any beam splitter},\ }\href {https://doi.org/10.1103/PhysRevA.89.062316} {\bibfield  {journal} {\bibinfo  {journal} {Physical Review A}\ }\textbf {\bibinfo {volume} {89}},\ \bibinfo {pages} {062316} (\bibinfo {year} {2014})}\BibitemShut {NoStop}%
\bibitem [{\citenamefont {Sawicki}(2016)}]{sawicki2015universality}%
  \BibitemOpen
  \bibfield  {author} {\bibinfo {author} {\bibfnamefont {A.}~\bibnamefont {Sawicki}},\ }\bibfield  {title} {\bibinfo {title} {Universality of beamsplitters},\ }\href {https://doi.org/10.26421/QIC16.3-4-6} {\bibfield  {journal} {\bibinfo  {journal} {Quantum Information and Computation}\ }\textbf {\bibinfo {volume} {16}},\ \bibinfo {pages} {291} (\bibinfo {year} {2016})}\BibitemShut {NoStop}%
\bibitem [{\citenamefont {Braunstein}\ and\ \citenamefont {Van~Loock}(2005)}]{braunstein2005quantum}%
  \BibitemOpen
  \bibfield  {author} {\bibinfo {author} {\bibfnamefont {S.~L.}\ \bibnamefont {Braunstein}}\ and\ \bibinfo {author} {\bibfnamefont {P.}~\bibnamefont {Van~Loock}},\ }\bibfield  {title} {\bibinfo {title} {Quantum information with continuous variables},\ }\href {https://doi.org/10.1103/RevModPhys.77.513} {\bibfield  {journal} {\bibinfo  {journal} {Reviews of modern physics}\ }\textbf {\bibinfo {volume} {77}},\ \bibinfo {pages} {513} (\bibinfo {year} {2005})}\BibitemShut {NoStop}%
\bibitem [{\citenamefont {Poulin}\ \emph {et~al.}(2011)\citenamefont {Poulin}, \citenamefont {Qarry}, \citenamefont {Somma},\ and\ \citenamefont {Verstraete}}]{poulin2011quantum}%
  \BibitemOpen
  \bibfield  {author} {\bibinfo {author} {\bibfnamefont {D.}~\bibnamefont {Poulin}}, \bibinfo {author} {\bibfnamefont {A.}~\bibnamefont {Qarry}}, \bibinfo {author} {\bibfnamefont {R.}~\bibnamefont {Somma}},\ and\ \bibinfo {author} {\bibfnamefont {F.}~\bibnamefont {Verstraete}},\ }\bibfield  {title} {\bibinfo {title} {Quantum simulation of time-dependent hamiltonians and the convenient illusion of hilbert space},\ }\href {https://doi.org/10.1103/PhysRevLett.106.170501} {\bibfield  {journal} {\bibinfo  {journal} {Physical review letters}\ }\textbf {\bibinfo {volume} {106}},\ \bibinfo {pages} {170501} (\bibinfo {year} {2011})}\BibitemShut {NoStop}%
\bibitem [{\citenamefont {Adesso}\ \emph {et~al.}(2014)\citenamefont {Adesso}, \citenamefont {Ragy},\ and\ \citenamefont {Lee}}]{adesso2014continuous}%
  \BibitemOpen
  \bibfield  {author} {\bibinfo {author} {\bibfnamefont {G.}~\bibnamefont {Adesso}}, \bibinfo {author} {\bibfnamefont {S.}~\bibnamefont {Ragy}},\ and\ \bibinfo {author} {\bibfnamefont {A.~R.}\ \bibnamefont {Lee}},\ }\bibfield  {title} {\bibinfo {title} {Continuous variable quantum information: Gaussian states and beyond},\ }\href {https://www.worldscientific.com/doi/abs/10.1142/S1230161214400010} {\bibfield  {journal} {\bibinfo  {journal} {Open Systems \& Information Dynamics}\ }\textbf {\bibinfo {volume} {21}},\ \bibinfo {pages} {1440001} (\bibinfo {year} {2014})}\BibitemShut {NoStop}%
\bibitem [{\citenamefont {Garc{\'i}a-Mart{\'i}n}\ \emph {et~al.}(2024)\citenamefont {Garc{\'i}a-Mart{\'i}n}, \citenamefont {Braccia},\ and\ \citenamefont {Cerezo}}]{garcia2024architectures}%
  \BibitemOpen
  \bibfield  {author} {\bibinfo {author} {\bibfnamefont {D.}~\bibnamefont {Garc{\'i}a-Mart{\'i}n}}, \bibinfo {author} {\bibfnamefont {P.}~\bibnamefont {Braccia}},\ and\ \bibinfo {author} {\bibfnamefont {M.}~\bibnamefont {Cerezo}},\ }\bibfield  {title} {\bibinfo {title} {Architectures and random properties of symplectic quantum circuits},\ }\href {https://arxiv.org/abs/2405.10264} {\bibfield  {journal} {\bibinfo  {journal} {arXiv preprint arXiv:2405.10264}\ } (\bibinfo {year} {2024})}\BibitemShut {NoStop}%
\bibitem [{\citenamefont {Motta}\ \emph {et~al.}(2020)\citenamefont {Motta}, \citenamefont {Sun}, \citenamefont {Tan}, \citenamefont {O’Rourke}, \citenamefont {Ye}, \citenamefont {Minnich}, \citenamefont {Brandao},\ and\ \citenamefont {Chan}}]{motta2020determining}%
  \BibitemOpen
  \bibfield  {author} {\bibinfo {author} {\bibfnamefont {M.}~\bibnamefont {Motta}}, \bibinfo {author} {\bibfnamefont {C.}~\bibnamefont {Sun}}, \bibinfo {author} {\bibfnamefont {A.~T.}\ \bibnamefont {Tan}}, \bibinfo {author} {\bibfnamefont {M.~J.}\ \bibnamefont {O’Rourke}}, \bibinfo {author} {\bibfnamefont {E.}~\bibnamefont {Ye}}, \bibinfo {author} {\bibfnamefont {A.~J.}\ \bibnamefont {Minnich}}, \bibinfo {author} {\bibfnamefont {F.~G.}\ \bibnamefont {Brandao}},\ and\ \bibinfo {author} {\bibfnamefont {G.~K.-L.}\ \bibnamefont {Chan}},\ }\bibfield  {title} {\bibinfo {title} {Determining eigenstates and thermal states on a quantum computer using quantum imaginary time evolution},\ }\href {https://doi.org/10.1038/s41567-019-0704-4} {\bibfield  {journal} {\bibinfo  {journal} {Nature Physics}\ }\textbf {\bibinfo {volume} {16}},\ \bibinfo {pages} {205} (\bibinfo {year} {2020})}\BibitemShut {NoStop}%
\bibitem [{\citenamefont {Zoufal}\ \emph {et~al.}(2021)\citenamefont {Zoufal}, \citenamefont {Lucchi},\ and\ \citenamefont {Woerner}}]{zoufal2021variational}%
  \BibitemOpen
  \bibfield  {author} {\bibinfo {author} {\bibfnamefont {C.}~\bibnamefont {Zoufal}}, \bibinfo {author} {\bibfnamefont {A.}~\bibnamefont {Lucchi}},\ and\ \bibinfo {author} {\bibfnamefont {S.}~\bibnamefont {Woerner}},\ }\bibfield  {title} {\bibinfo {title} {Variational quantum boltzmann machines},\ }\href {https://doi.org/10.1007/s42484-020-00033-7} {\bibfield  {journal} {\bibinfo  {journal} {Quantum Machine Intelligence}\ }\textbf {\bibinfo {volume} {3}},\ \bibinfo {pages} {1} (\bibinfo {year} {2021})}\BibitemShut {NoStop}%
\bibitem [{\citenamefont {Zoufal}\ \emph {et~al.}(2023)\citenamefont {Zoufal}, \citenamefont {Sutter},\ and\ \citenamefont {Woerner}}]{zoufal2021error}%
  \BibitemOpen
  \bibfield  {author} {\bibinfo {author} {\bibfnamefont {C.}~\bibnamefont {Zoufal}}, \bibinfo {author} {\bibfnamefont {D.}~\bibnamefont {Sutter}},\ and\ \bibinfo {author} {\bibfnamefont {S.}~\bibnamefont {Woerner}},\ }\bibfield  {title} {\bibinfo {title} {Error bounds for variational quantum time evolution},\ }\href {https://doi.org/10.1103/PhysRevApplied.20.044059} {\bibfield  {journal} {\bibinfo  {journal} {Physical Review Applied}\ }\textbf {\bibinfo {volume} {20}},\ \bibinfo {pages} {044059} (\bibinfo {year} {2023})}\BibitemShut {NoStop}%
\bibitem [{\citenamefont {Silva}\ \emph {et~al.}(2023)\citenamefont {Silva}, \citenamefont {Taddei}, \citenamefont {Carrazza},\ and\ \citenamefont {Aolita}}]{silva2023fragmented}%
  \BibitemOpen
  \bibfield  {author} {\bibinfo {author} {\bibfnamefont {T.~L.}\ \bibnamefont {Silva}}, \bibinfo {author} {\bibfnamefont {M.~M.}\ \bibnamefont {Taddei}}, \bibinfo {author} {\bibfnamefont {S.}~\bibnamefont {Carrazza}},\ and\ \bibinfo {author} {\bibfnamefont {L.}~\bibnamefont {Aolita}},\ }\bibfield  {title} {\bibinfo {title} {Fragmented imaginary-time evolution for early-stage quantum signal processors},\ }\href {https://doi.org/10.1038/s41598-023-45540-2} {\bibfield  {journal} {\bibinfo  {journal} {Scientific Reports}\ }\textbf {\bibinfo {volume} {13}},\ \bibinfo {pages} {18258} (\bibinfo {year} {2023})}\BibitemShut {NoStop}%
\bibitem [{\citenamefont {Childs}\ and\ \citenamefont {Wiebe}(2012)}]{childs2012hamiltonian}%
  \BibitemOpen
  \bibfield  {author} {\bibinfo {author} {\bibfnamefont {A.~M.}\ \bibnamefont {Childs}}\ and\ \bibinfo {author} {\bibfnamefont {N.}~\bibnamefont {Wiebe}},\ }\bibfield  {title} {\bibinfo {title} {Hamiltonian simulation using linear combinations of unitary operations},\ }\href {https://arxiv.org/abs/1202.5822} {\bibfield  {journal} {\bibinfo  {journal} {arXiv preprint arXiv:1202.5822}\ } (\bibinfo {year} {2012})}\BibitemShut {NoStop}%
\bibitem [{\citenamefont {Nielsen}\ and\ \citenamefont {Chuang}(2000)}]{nielsen2000quantum}%
  \BibitemOpen
  \bibfield  {author} {\bibinfo {author} {\bibfnamefont {M.~A.}\ \bibnamefont {Nielsen}}\ and\ \bibinfo {author} {\bibfnamefont {I.~L.}\ \bibnamefont {Chuang}},\ }\href@noop {} {\emph {\bibinfo {title} {Quantum Computation and Quantum Information}}}\ (\bibinfo  {publisher} {Cambridge University Press},\ \bibinfo {address} {Cambridge},\ \bibinfo {year} {2000})\BibitemShut {NoStop}%
\bibitem [{\citenamefont {Efthymiou}\ \emph {et~al.}(2021)\citenamefont {Efthymiou}, \citenamefont {Ramos-Calderer}, \citenamefont {Bravo-Prieto}, \citenamefont {P{\'e}rez-Salinas}, \citenamefont {Garc{\'\i}a-Mart{\'\i}n}, \citenamefont {Garcia-Saez}, \citenamefont {Latorre},\ and\ \citenamefont {Carrazza}}]{efthymiou2020qibo}%
  \BibitemOpen
  \bibfield  {author} {\bibinfo {author} {\bibfnamefont {S.}~\bibnamefont {Efthymiou}}, \bibinfo {author} {\bibfnamefont {S.}~\bibnamefont {Ramos-Calderer}}, \bibinfo {author} {\bibfnamefont {C.}~\bibnamefont {Bravo-Prieto}}, \bibinfo {author} {\bibfnamefont {A.}~\bibnamefont {P{\'e}rez-Salinas}}, \bibinfo {author} {\bibfnamefont {D.}~\bibnamefont {Garc{\'\i}a-Mart{\'\i}n}}, \bibinfo {author} {\bibfnamefont {A.}~\bibnamefont {Garcia-Saez}}, \bibinfo {author} {\bibfnamefont {J.~I.}\ \bibnamefont {Latorre}},\ and\ \bibinfo {author} {\bibfnamefont {S.}~\bibnamefont {Carrazza}},\ }\bibfield  {title} {\bibinfo {title} {Qibo: a framework for quantum simulation with hardware acceleration},\ }\href {https://doi.org/10.1088/2058-9565/ac39f5} {\bibfield  {journal} {\bibinfo  {journal} {Quantum Science and Technology}\ }\textbf {\bibinfo {volume} {7}},\ \bibinfo {pages} {015018} (\bibinfo {year} {2021})}\BibitemShut {NoStop}%
\bibitem [{\citenamefont {Efthymiou}\ \emph {et~al.}(2022)\citenamefont {Efthymiou}, \citenamefont {Lazzarin}, \citenamefont {Pasquale},\ and\ \citenamefont {Carrazza}}]{efthymiou2022quantum}%
  \BibitemOpen
  \bibfield  {author} {\bibinfo {author} {\bibfnamefont {S.}~\bibnamefont {Efthymiou}}, \bibinfo {author} {\bibfnamefont {M.}~\bibnamefont {Lazzarin}}, \bibinfo {author} {\bibfnamefont {A.}~\bibnamefont {Pasquale}},\ and\ \bibinfo {author} {\bibfnamefont {S.}~\bibnamefont {Carrazza}},\ }\bibfield  {title} {\bibinfo {title} {Quantum simulation with just-in-time compilation},\ }\href {https://doi.org/10.22331/q-2022-09-22-814} {\bibfield  {journal} {\bibinfo  {journal} {Quantum}\ }\textbf {\bibinfo {volume} {6}},\ \bibinfo {pages} {814} (\bibinfo {year} {2022})}\BibitemShut {NoStop}%
\bibitem [{\citenamefont {Rudolph}\ and\ \citenamefont {Grover}(2002)}]{rudolph20022}%
  \BibitemOpen
  \bibfield  {author} {\bibinfo {author} {\bibfnamefont {T.}~\bibnamefont {Rudolph}}\ and\ \bibinfo {author} {\bibfnamefont {L.}~\bibnamefont {Grover}},\ }\bibfield  {title} {\bibinfo {title} {A 2 rebit gate universal for quantum computing},\ }\bibfield  {journal} {\bibinfo  {journal} {arXiv preprint quant-ph/0210187}\ }\href {https://doi.org/10.48550/arXiv.quant-ph/0210187} {10.48550/arXiv.quant-ph/0210187} (\bibinfo {year} {2002})\BibitemShut {NoStop}%
\bibitem [{\citenamefont {Shi}(2002)}]{shi2002both}%
  \BibitemOpen
  \bibfield  {author} {\bibinfo {author} {\bibfnamefont {Y.}~\bibnamefont {Shi}},\ }\bibfield  {title} {\bibinfo {title} {Both toffoli and controlled-not need little help to do universal quantum computation},\ }\href {https://arxiv.org/abs/quant-ph/0205115} {\bibfield  {journal} {\bibinfo  {journal} {arXiv preprint quant-ph/0205115}\ } (\bibinfo {year} {2002})}\BibitemShut {NoStop}%
\bibitem [{\citenamefont {Aharonov}(2003)}]{aharonov2003simple}%
  \BibitemOpen
  \bibfield  {author} {\bibinfo {author} {\bibfnamefont {D.}~\bibnamefont {Aharonov}},\ }\bibfield  {title} {\bibinfo {title} {A simple proof that toffoli and hadamard are quantum universal},\ }\href {https://arxiv.org/abs/quant-ph/0301040} {\bibfield  {journal} {\bibinfo  {journal} {arXiv preprint quant-ph/0301040}\ } (\bibinfo {year} {2003})}\BibitemShut {NoStop}%
\bibitem [{\citenamefont {Holmes}\ \emph {et~al.}(2023)\citenamefont {Holmes}, \citenamefont {Coble}, \citenamefont {Sornborger},\ and\ \citenamefont {Suba{\c{s}}{\i}}}]{holmes2021nonlinear}%
  \BibitemOpen
  \bibfield  {author} {\bibinfo {author} {\bibfnamefont {Z.}~\bibnamefont {Holmes}}, \bibinfo {author} {\bibfnamefont {N.}~\bibnamefont {Coble}}, \bibinfo {author} {\bibfnamefont {A.~T.}\ \bibnamefont {Sornborger}},\ and\ \bibinfo {author} {\bibfnamefont {Y.}~\bibnamefont {Suba{\c{s}}{\i}}},\ }\bibfield  {title} {\bibinfo {title} {On nonlinear transformations in quantum computation},\ }\href {https://doi.org/10.1103/PhysRevResearch.5.013105} {\bibfield  {journal} {\bibinfo  {journal} {Physical Review Research}\ }\textbf {\bibinfo {volume} {5}},\ \bibinfo {pages} {013105} (\bibinfo {year} {2023})}\BibitemShut {NoStop}%
\end{thebibliography}%
\clearpage

\appendix
\setcounter{theorem}{0}
\setcounter{corollary}{0}
\setcounter{proposition}{0}
\setcounter{problem}{0}
\setcounter{definition}{0}

\onecolumngrid

\section*{Supplemental Information for ``Gate-based quantum simulation of Gaussian bosonic circuits on exponentially many modes''}

\section{Framework}
\label{app:framework}

\subsection{Time evolution of quadrature operators in phase space}

Let us consider a system of bosons with $M$ modes, and let us assume $\hbar=1$. The state space of such a system is 
$    \HC = \bigotimes_{m=1}^M \FC_m$, where $\FC_m$ is the Fock space associated to the $m$-th mode. That is, each $\FC_m$ is spanned by the infinitely-many basis vectors $\{\ket{k}_m\}_{k\in\mathbb{N}}$ indicating the occupancy number of the mode. For example, if we fix the number of bosons to $3$ and let the number of modes be $M=4$, the basis state $\ket{0,1,2,0}\equiv\ket{0}\otimes\ket{1}\otimes\ket{2}\otimes\ket{0}$ corresponds to the state with one boson occupying the second mode and two bosons occupying the third mode. These basis states are eigenstates of the particle number operator, since they have a fixed number of particles, and are known as Fock states. 
The (non-Hermitian) creation $\hat{a}_m^\dagger$ and annihilation $\hat{a}_m$ operators are defined by their action on Fock states as follows,
\begin{equation}\begin{split}
    &\hat{a}_m^\dagger \ket{N_0,\dots, N_m,\dots,N_M} = \sqrt{N_m+1} \,\ket{N_0,\dots, N_m+1,\dots,N_M}\,, \\
    &\hat{a}_m \ket{N_0,\dots, N_m,\dots,N_M} = \sqrt{N_m}\, \ket{N_0,\dots, N_m-1,\dots,N_M}\,.
\end{split}
\end{equation}
They satisfy the commutation relations
\begin{equation} \begin{split}
    &[\hat{a}_m^\dagger, \hat{a}_{m'}^\dagger] = [\hat{a}_m, \hat{a}_{m'}] =0 \,,\\
    &[\hat{a}_m, \hat{a}_{m'}^\dagger] = \delta_{m,m'}\,.
\end{split}
\end{equation}
The (Hermitian) particle number operator is given by $\hat{a}_m^\dagger \hat{a}_m$.
The (Hermitian) position $\hat{q}_m$ and momentum $\hat{p}_m$ operators can be defined from the creation and annihilation operators as
\begin{equation}
    \opq_m = \frac{\hat{a}_m + \opa_m^\dagger}{\sqrt{2}}\,, \qquad \opp_m = i\,\frac{\opa_m^\dagger-\opa_m }{\sqrt{2}}\,.
\end{equation}
Which in turn conversely implies that
\begin{equation} \label{eq-ap:a_m-to-q_m}
    \opa_m = \frac{\opq_m +i \opp_m}{\sqrt{2}}\,, \qquad \opa_m^\dagger = \frac{\opq_m - i \opp_m}{\sqrt{2}}\,.
\end{equation}
The corresponding commutation relations for position and momentum are
\begin{equation}\label{eq-ap:p-q-comm}
\begin{split}
&[\opq_m, \opq_{m'}]= [\opp_m, \opp_{m'}] =0\,,\\
&[\opq_m, \opp_{m'}]=i\delta_{m,m'}\,.
\end{split}
\end{equation}

Let us collect the positions and momenta in a vector $\hat{\vec{z}}=(\opq_1,\dots,\opq_M,\opp_1,\dots,\opp_M)^T$. This allows us to write the commutation relations in Eq.~\eqref{eq-ap:p-q-comm} as
\begin{equation}
    \left[\z,\z^T\right] = i \Omega\,,
\end{equation}
where the $\Omega$ matrix is
\begin{equation}
    \Omega = \begin{pmatrix} 0& \id_{M} \\ - \id_{M} & 0\end{pmatrix} = iY\otimes \id_M \,.
\end{equation}
Here, $Y$ is the usual $2\times 2$ Pauli matrix and $\id_M$ the $M\times M$ identity matrix.
It will be convenient for us to explicitly write down the matrix entries of $\Omega$, which are
\begin{equation}
    \Omega_{\gamma,\gamma'} = \delta_{\gamma, \gamma'-M} -\delta_{\gamma, \gamma'+M}\,.
\end{equation}

Let us assume that our quantum bosonic system is governed by a quadratic time-independent Hamiltonian of the form
\begin{equation}
    \hat{H} = \frac{1}{2} \vec{\z}^T K \vec{\z}\,,
\end{equation}
where $K$ is a real $2M\times 2M$ symmetric matrix. In the Heisenberg picture, the equation of motion of an observable $\hat{O}$ is
\begin{equation}
    \frac{\partial \hat{O}}{\partial t} = i[\hat{H},\hat{O}]\,.
\end{equation}
Therefore, we find
\begin{equation}  \begin{split}
    \frac{\partial \z_\gamma}{\partial t}  &= i [\hat{H}, \z_\gamma] = \frac{i}{2}\Big[\sum_{\alpha\beta} \z_\alpha K_{\alpha\beta} \z_\beta, \z_\gamma\Big]  = \frac{i}{2}\sum_{\alpha\beta} K_{\alpha\beta} \Big[\z_\alpha  \z_\beta, \z_\gamma\Big] = \frac{i}{2}\sum_{\alpha\beta} K_{\alpha\beta} \left(\z_\alpha \Big[ \z_\beta, \z_\gamma\Big]+ \Big[ \z_\alpha, \z_\gamma\Big] \z_\beta \right) \\ &= \frac{i}{2}\sum_{\alpha\beta} K_{\alpha\beta} \left(\z_\alpha i(\delta_{\beta,\gamma-M} - \delta_{\beta,\gamma+M})+ i(\delta_{\alpha,\gamma-M} - \delta_{\alpha,\gamma+M} ) \z_\beta \right) \\ & = \frac{1}{2}\sum_{\alpha\beta} K_{\alpha\beta} \left(-\z_\alpha \Omega_{\beta\gamma}-\Omega_{\alpha\gamma} \z_\beta \right) = \frac{1}{2} \sum_{\alpha\beta} \left(\Omega_{\gamma\beta} K_{\beta \alpha} \z_\alpha + \Omega_{\gamma\alpha} K_{\alpha\beta} \z_\beta \right) \\ &= \left(\Omega K \z\right)_\gamma\,,
\end{split}
\end{equation}
where we used  $[AB,C]= A[B,C] + [A,C]B$, and the fact that $K$ is symmetric and $\Omega$ anti-symmetric. Hence, we arrive at
\begin{equation} \label{eq-ap:motion}
    \frac{\partial \z}{\partial t} = \Omega K \z\,.
\end{equation}
The solution to this differential equation is given by
\begin{equation}
    \z(t) = e^{t \Omega K} \z(0)\,.
\end{equation}
The solution of Eq.~\eqref{eq-ap:motion} must preserve the commutation relations of Eq.~\eqref{eq-ap:p-q-comm} in order to leave the kinematics invariant. 
Let us call $Q(t)=e^{t \Omega K}$ the propagator that takes the vector $\z(0)$ at time $0$ to the vector $\z(t)$ at time $t$. 
We impose $\left[\z(t),\z(t)^T\right] = [\z(0),\z(0)^T]$, which leads to
\begin{equation} \begin{split}
    \left[Q(t)\z(0),(Q(t)\z(0))^T\right]_{\alpha\beta} &=  \left[\left(\sum_\gamma Q(t)_{\alpha \gamma} \z_\gamma\right),\left(\sum_{\gamma'} Q(t)_{\beta \gamma'} \z_{\gamma'}\right)\right] = \sum_{\gamma,\gamma'} Q(t)_{\alpha \gamma} Q(t)_{\beta \gamma'} \left[\z_\gamma, \z_{\gamma'} \right] \\ &=i \sum_{\gamma,\gamma'} Q(t)_{\alpha \gamma} Q(t)_{\beta \gamma'} (\delta_{\gamma, \gamma'-M} -\delta_{\gamma, \gamma'+M}) =  i \sum_{\gamma,\gamma'} Q(t)_{\alpha \gamma} \Omega_{\gamma,\gamma'} Q(t)^T_{ \gamma' \beta} \\ &= i \left(Q(t)\, \Omega \,Q^T(t) \right)_{\alpha\beta}  \,. 
\end{split}
\end{equation}
It is clear then that $Q(t)\,\Omega \,Q^T(t)=\Omega$ if we are to maintain the commutation relations between positions and momenta. This is precisely the defining condition of symplectic matrices. Hence, the time evolution of $\z$ in phase space is given by a $2M\times 2M$ real symplectic matrix belonging to the group $\mathbb{SP}(M,\mathbb{R})$.

Let us now address the question of how the expectation value $\langle \z\rangle=(\la \opq_1\ra,\dots,\la \opq_M\ra ,\la \opp_1\ra,\dots,\la \opp_M\ra )^T $ of $\z$ evolves in time, given an initial quantum state $\rho_0$. We find that
\small
\begin{equation} \label{eq-ap:exp-val-motion}
    \frac{\partial \langle \z_\gamma \rangle}{\partial t} = \frac{\partial\, \Tr[\z_\gamma \rho_0]}{\partial t} = \Tr\left[\frac{\partial\z_\gamma}{\partial t} \rho_0 \right] = \Tr\left[(\Omega K \z)_\gamma \rho_0 \right] = \Tr\left[\left(\sum_{\alpha\beta} \Omega_{\gamma\alpha}  K_{\alpha\beta} \z_\beta\right) \rho_0 \right] = \sum_{\alpha\beta} \Omega_{\gamma\alpha}  K_{\alpha\beta}  \Tr\left[ \z_\beta \rho_0 \right] = (\Omega K\la\z\ra)_\gamma\,,
\end{equation}
\normalsize
which implies
\begin{equation}
     \frac{\partial \langle \z \rangle}{\partial t} =  \Omega K \la\z\ra \,,
    \label{eq-ap:equa_diff}
\end{equation}
and
\begin{equation} \label{eq-ap:phase-space-ev}
    \la\z\ra(t) = e^{t \Omega K} \la\z\ra(0)\,.
\end{equation}

We can also collect the expectation value of products of quadrature operators over $\rho_0$ in the $2M\times 2M$ positive-definite covariance matrix $\s$ whose entries are given by  
\begin{equation}
    \s_{\alpha \beta}=\frac{1}{2}\langle \z_\alpha \z_\beta+ \z_\beta\z_\alpha\rangle -\langle \z_\alpha\rangle\langle \z_\beta\rangle\,.
\end{equation}

 Let us derive the corresponding equation of motion. We start by considering
\small
\begin{align}
    \frac{\partial (\z_\gamma \z_\delta)}{\partial t} &= i [\hat{H}, \z_\gamma\z_\delta] = \frac{i}{2}\Big[\sum_{\alpha\beta} \z_\alpha K_{\alpha\beta} \z_\beta, \z_\gamma \z_\delta \Big] = \frac{i}{2}\sum_{\alpha\beta} K_{\alpha\beta} \Big[\z_\alpha  \z_\beta, \z_\gamma\z_\delta \Big] \nonumber \\ & = \frac{i}{2}\sum_{\alpha\beta} K_{\alpha\beta} \left(\z_\alpha \Big[ \z_\beta, \z_\gamma\Big]\z_\delta + \Big[ \z_\alpha, \z_\gamma\Big] \z_\beta \z_\delta + \z_\gamma \z_\alpha \Big[ \z_\beta, \z_\delta\Big]   + \z_\gamma\Big[ \z_\alpha, \z_\delta\Big] \z_\beta  \right) \nonumber \\ &= -\frac{1}{2} \sum_{\alpha\beta} K_{\alpha\beta} \left(\z_\alpha (\delta_{\beta,\gamma-M} - \delta_{\beta,\gamma+M})\z_\delta + (\delta_{\alpha,\gamma-M} - \delta_{\alpha,\gamma+M}) \z_\beta \z_\delta + \z_\gamma \z_\alpha (\delta_{\beta,\delta-M} - \delta_{\beta,\delta+M})   + \z_\gamma (\delta_{\alpha,\delta-M} - \delta_{\alpha,\delta+M}) \z_\beta  \right) \nonumber \\ &= -\frac{1}{2} \sum_{\alpha\beta} K_{\alpha\beta} \left(\z_\alpha \Omega_{\beta\gamma} \z_\delta + \Omega_{\alpha\gamma} \z_\beta \z_\delta + \z_\gamma \z_\alpha \Omega_{\beta\delta}  + \z_\gamma \Omega_{\alpha\delta} \z_\beta  \right) \nonumber \\ &= -  (\z^T K \Omega)_\gamma \z_\delta -   \z_\gamma (\z^T K \Omega )_\delta  \nonumber = ( \Omega K \z)_\gamma \z_\delta +  \z_\gamma ( \Omega K\z)_\delta    \,,
\end{align}
\normalsize
where we used that $[AB,CD]= A[B,C]D + [A,C]BD + CA[B,D] + C[A,D]B$, and the fact that $K$ ($\Omega$) is symmetric (anti-symmetric). In matrix form, we get
\begin{equation}
    \frac{\partial (\z \z^T)}{ \partial t} = \Omega K \z \z^T - \z \z^T K\Omega \,.
\end{equation}
Let us now look at the evolution of the expectation value of two-point correlators. Analogously to Eq.~\eqref{eq-ap:exp-val-motion}, we have
\begin{align}
    \frac{\partial \langle \z_\gamma\z_\delta \rangle}{\partial t} &= \frac{\partial\, \Tr[\z_\gamma \z_\delta\rho]}{\partial t} = \Tr\left[\frac{\partial\z_\gamma}{\partial t}\z_\delta \rho + \z_\gamma \frac{\partial\z_\delta}{\partial t} \rho \right] = \Tr\left[(\Omega K \z)_\gamma \z_\delta \rho + \z_\gamma (\Omega K \z)_\delta \rho \right] \nonumber \\& = \Tr\left[\left(\sum_{\alpha\beta} \Omega_{\gamma\alpha}  K_{\alpha\beta} \z_\beta\right)\z_\delta \rho + \z_\gamma\left(\sum_{\alpha\beta} \Omega_{\delta\alpha}  K_{\alpha\beta} \z_\beta \right) \rho\right] = \sum_{\alpha\beta} \Omega_{\gamma\alpha}  K_{\alpha\beta}  \Tr\left[ \z_\beta\z_\delta \rho \right] + \sum_{\alpha\beta} \Omega_{\delta\alpha}  K_{\alpha\beta}  \Tr\left[ \z_\gamma\z_\beta \rho \right]\nonumber \\&= (\Omega K \la\z \z^T\ra)_{\gamma\delta} - ( \la\z \z^T\ra  K \Omega)_{\gamma\delta}\,,
\end{align}
or, in matrix form,
\begin{equation}
    \frac{\partial (\la \z \z^T \ra)}{ \partial t} = \Omega K \la \z \z^T\ra  - \la \z \z^T\ra K\Omega  \,.
\end{equation}
Therefore, we arrive at    
\begin{equation}
    \frac{\partial \s}{ \partial t} = \Omega K \s-\s K\Omega\,.
\end{equation}
If $K$ represents a particle-preserving Hamiltonian, then $[\Omega, K]=0$ according to Supplemental Proposition~\ref{prop-ap:unitary-ps}, so we can write
\begin{equation}
    \frac{\partial \s}{ \partial t} = \left[\Omega K, \s\right]\,.
\end{equation}
The solution to this equation is
\begin{equation}
    \s =  e^{\Omega K t} \,\s(0)\,  e^{-\Omega K t}\,.
\end{equation}
If instead $K$ represents a non-particle-preserving Hamiltonian such that $\{\Omega, K\}=0$, we can write 
\begin{equation}
    \frac{\partial \s}{ \partial t} = \left\{\Omega K, \s\right\}\,,
\end{equation}
whose solution is
\begin{equation}
    \s(t) =  e^{\Omega K t} \,\s(0)\,  e^{\Omega K t}\,.
\end{equation}

\subsection{Pauli basis for the symplectic algebra}
\label{app:symplectic}
Here we present a useful Supplemental Proposition that provides a Pauli basis for the symplectic algebra $\mathfrak{sp}(M,\mathbb{R})$.

\begin{supproposition}
    \label{lem-ap:sp-algebra}
     An orthogonal basis for the standard representation of the $\mathfrak{sp}(M,\mathbb{R})$ algebra, where $M=2^{n}$, is given by the set 
     \begin{equation} \label{eq-ap:sp-dla}
	   B_{\mathfrak{sp}(M,\mathbb{R})}\equiv i\{Y\otimes P_s \}\,\cup\, i\{\id\otimes P_a\} \,\cup\,\{X\otimes P_s \}\,\cup\,\{Z\otimes P_s \} \,,
    \end{equation}
    where $P_s$ and $P_a$ belong to the sets of arbitrary symmetric and anti-symmetric Pauli strings on $n$ qubits, respectively, and $\id,X,Y,Z$ are the usual $2\times 2$ Pauli matrices.
\end{supproposition}

\begin{proof}
    We first recall that any $2M\times 2M$ matrix $A$ satisfying $A^T\Omega=-\Omega A$ belongs to $\mathfrak{sp}(M,\mathbb{R})\subset {\rm End}(\mbb{R}^{2M})$. Clearly, the (phased) Pauli operators on $n+1$ qubits in Eq.~\eqref{eq-ap:sp-dla} constitute $2^{n+1}\times 2^{n+1} =2 M \times 2M$ matrices, with $M=2^n$. Also, note they are all real-valued (although not all anti-Hermitian), which follows from the fact that a Pauli string is real (purely imaginary) when it contains an even (odd) number of $Y$'s, i.e. when it is symmetric (anti-symmetric). Thus, $B_{\mathfrak{sp}(M,\mathbb{R})}\subset {\rm End}(\mbb{R}^{2M})$. We know from Proposition 1 in Ref.~\cite{garcia2024architectures} that they satisfy the symplectic property, implying $B_{\mathfrak{sp}(M,\mathbb{R})}\subset \spf$. Given they are Hilbert-Schmidt orthogonal, and that $|B_{\mathfrak{sp}(M,\mathbb{R})}|=M(2M+1)=\dim(\spf)$, to prove they constitute an orthogonal basis it simply remains to show that they are closed under commutation. 
    
    To do so, we first recall that the commutator of two anti-symmetric or symmetric matrices is anti-symmetric, whereas the commutator of a symmetric matrix and an anti-symmetric one is symmetric.
     We start with the (non-zero) commutator $X\otimes P_s$ and $X\otimes P_s'$, which gives an operator of the following form
     \begin{equation}
         [X\otimes P_s , X\otimes P_s'] \propto\pm i\id\otimes P_a\,.
     \end{equation}
     The $\pm i$ factor follows from the fact that the Pauli strings $P_s$ and $P_s'$ differ at an odd number of sites. The same is true if we replace $X$ by $Z$ on the first qubit, 
     \begin{equation}
         [Z\otimes P_s , Z\otimes P_s'] \propto\pm i\id\otimes P_a\,.
     \end{equation}
     We continue by computing the (non-zero) commutator of two operators of the form $X\otimes P_s$ and $Z\otimes P_s'$,
     \begin{equation}
         [X\otimes P_s , Z\otimes P_s'] \propto\pm iY\otimes P_s''\,.
     \end{equation}
     Again, the $\pm i$ factor follows from the fact that $X\otimes P_s$ and $Z\otimes P_s'$ differ at an odd number of sites. Next, let us look at the non-zero commutator of operators $X\otimes P_s$ or $Z\otimes P_s$ with $i\id\otimes P_a$,
     \begin{align}
         [X\otimes P_s,i\id\otimes P_a] \propto\pm X \otimes P_s'\quad{\rm or} \quad
         [Z\otimes P_s,i\id\otimes P_a] \propto\pm Z \otimes P_s'\,.
     \end{align}
     Here, the $i$ factor arising from commuting the Pauli strings cancels out with the $i$ in $i\id\otimes P_a$. Similarly, the non-zero commutator of $X\otimes P_s$ or $Z\otimes P_s$ with $iY\otimes P_s'$ is as follows
     \begin{align}
         [X\otimes P_s,iY\otimes P_s'] \propto\pm Z\otimes P_s''\quad {\rm or}\quad
         [Z\otimes P_s,iY\otimes P_s'] \propto\pm X\otimes P_s''\,.
     \end{align}
     Furthermore, commuting $iY\otimes P_s$ and $iY\otimes P_s'$, or $i\id\otimes P_a$ with $i\id\otimes P_a'$ leads to either zero or an operator of the form
     \begin{align}
         [iY\otimes P_s,iY\otimes P_s'] \propto\pm i\id\otimes P_a\quad {\rm or} \quad
         [i\id\otimes P_a,i\id\otimes P_a'] \propto\pm i\id\otimes P_a''\,.
     \end{align}
     Finally, the non-zero commutator of $iY\otimes P_s$ and  $i\id\otimes P_a$ gives $\pm i Y\otimes P_s $, 
     \begin{equation}
         [iY\otimes P_s,i\id\otimes P_a] \propto\pm i Y\otimes P_s \,.
     \end{equation}
     Therefore, we conclude that the set $B_{\spf}\subset \spf$ of mutually orthogonal operators is closed under commutation and satisfies $\dim(B_{\spf})=M(2M+1) =\dim(\spf)$. Thus, it constitutes an orthogonal basis for the standard representation of $\spf$.
\end{proof}    

\subsection{Particle-preserving gates}
\label{app:prop1}
Next, we show that particle-preserving Gaussian bosonic (GB) gates lead to unitary evolutions at the qubit level.

\begin{supproposition} \label{prop-ap:unitary-ps}
    When a gate  generator is of the form
    \begin{equation} \label{eq-ap:parti-preserv-H}
    \hat{H} = \sum_{m,m'=1}^{M} h_{m m'} \,\hat{a}_m^\dagger \hat{a}_{m'} +\frac{\Tr[h]}{2}\id_{2M}  \,, 
\end{equation}
where $h$ is a Hermitian matrix, then $[\Omega, K]=0$, and the propagator $e^{t\Omega K}$ is the real time evolution under the effective Hamiltonian $-i\Omega K$.
\end{supproposition}

\begin{proof}
    Expressed in terms of position and momentum operators, we find
\begin{align} \label{eq-ap:H-particle-preserving}
    \hat{H} &= \frac{1}{2} \sum_{m,m'=1}^{2^n} h_{m m'}\, (q_m -i p_m) (q_{m'} + i p_{m'})  +\frac{\Tr[h]}{2}\id_{2M} \\&=
    \frac{1}{2} \sum_{m,m'=1}^{2^n} h_{m m'} \, (q_m q_{m'} + p_m p_{m'} + i q_m p_{m'} - i p_m q_{m'})  +\frac{\Tr[h]}{2}\id_{2M} \nonumber \\& = \frac{1}{2} \sum_{m=1}^{2^n} h_{mm} \, (q_m^2 + p_m^2)  +  \sum_{\substack{m,m'=1 \\m'>m}}^{2^n} {\rm Re}[h_{mm'}] \, (q_m q_{m'} + p_m p_{m'}) +  \sum_{\substack{m,m'=1 \\m'>m}}^{2^n} {\rm Im}[h_{mm'}] \, (q_m p_{m'} - p_m q_{m'}) \,,
\end{align}
where we used the commutation relations from Eq.~\eqref{eq-ap:p-q-comm}. In other words, particle-preserving gate generators are such that the $K$ matrix is a real linear combination of Paulis of the form $\id\otimes P_s$ (corresponding to the first two sums in Eq.~\eqref{eq-ap:H-particle-preserving}) and/or $Y\otimes P_a$ (corresponding to the last sum in Eq.~\eqref{eq-ap:H-particle-preserving}). This automatically implies that $[\Omega,K]=0$. Finally, either $\Omega K$ is a real combination of Paulis of the form $ (iY\otimes \id_M) (\id\otimes P_s)= iY \otimes P_s$ or $ (iY\otimes \id_M) (Y\otimes P_a)= i\id \otimes P_a$ (both of which are in the symplectic algebra according to~\Cref{lem-ap:sp-algebra}). That is, $\Omega K$ is anti-Hermitian and the symplectic propagator $e^{t\Omega K}$ is unitary.
Hence, these types of gate generators result in unitary dynamics in phase space.
\end{proof}

\subsection{Non-particle-preserving gates}

We here show that a family of non-particle-preserving GB gates lead to an imaginary-time evolution at the qubit level.

\begin{supproposition} \label{prop-ap:imaginary-ps}
    When a gate  generator is of the form
    \begin{equation} \label{eq:no-parti-preserv-H}
    \hat{H} = \sum_{m,m'=1}^{M} \Delta_{m m'}^\dagger \,\hat{a}_m \hat{a}_{m'}  +\sum_{m,m'=1}^{M} \Delta_{m m'} \,\hat{a}_m^\dagger \hat{a}_{m'}^\dagger \,,
\end{equation}
where $\Delta$ is a symmetric matrix, then $\{\Omega, K\}=0$, and the propagator $e^{t\Omega K}$ is an imaginary time evolution under the effective Hamiltonian $-\Omega K$. 
\end{supproposition}

\begin{proof}
    In terms of positions and momenta, we find
\begin{align} \label{eq-ap:H-no-particle-preserving}
    \hat{H} &= \frac{1}{2} \sum_{m,m'=1}^{2^n} \Delta_{mm'}^\dagger \, (q_m +i p_m) (q_{m'} + i p_{m'}) + \frac{1}{2} \sum_{m,m'=1}^{2^n} \Delta_{mm'}\, (q_m -i p_m) (q_{m'} - i p_{m'}) \nonumber\\ &= \frac{1}{2} \sum_{m,m'=1}^{2^n} \Delta_{mm'}^\dagger \, (q_m q_{m'} - p_m p_{m'} + i q_m p_{m'} + i p_m q_{m'}) + \frac{1}{2} \sum_{m,m'=1}^{2^n} \Delta_{mm'} \, (q_m q_{m'} - p_m p_{m'} - i q_m p_{m'} - i p_m q_{m'}) \nonumber \\ &=  \sum_{m,m'=1}^{2^n} {\rm Re} \left[\Delta_{mm'}\right] \, (q_m q_{m'} - p_m p_{m'}) + \sum_{m,m'=1}^{2^n}{\rm Im}\left[  \Delta_{mm'}\right] \, (q_m p_{m'} +  p_m q_{m'})\,.
\end{align}
In this case, the $K$ matrix is a real linear combination of Paulis of the form $Z\otimes P_s$ (corresponding to the first sum in Eq.~\eqref{eq-ap:H-no-particle-preserving}) and/or $X\otimes P_s$ (corresponding to the second sum in Eq.~\eqref{eq-ap:H-no-particle-preserving}). 
This implies that $\{\Omega,K\}=0$.
Then, either $\Omega K$ is a real combination of Paulis of the form $ (iY\otimes \id_M) (Z\otimes P_s)= X \otimes P_s$ or $ (iY\otimes \id_M) (X\otimes P_s)= Z\otimes P_s$ (both of which are in the symplectic algebra according to~\Cref{lem-ap:sp-algebra}). That is, $\Omega K$ is Hermitian and the symplectic propagator $e^{t\Omega K}$ is given by the imaginary-time evolution of the effective gate generator $-\Omega K$.
\end{proof}

\section{From bosonic gates to qubit gates}
\label{app:gates}

\subsection{Local gates}

\begin{itemize}
    \item \textbf{Phase gate:}
    This gate is described by the generator $\hat{H} = \opq_m^2 + \opp_m^2$ in terms of bosonic operators. Therefore the real symmetric $K$ matrix can be expressed as
    \begin{equation}
        K=2 \left(\ketbra{m}{m} + \ketbra{m+M}{m+M}\right) =  2 \left(\ketbra{0}{0} \otimes \ketbra{m}{m} + \ketbra{1}{1} \otimes \ketbra{m}{m} \right)= 2 \id \otimes \ketbra{m}{m}\,.
    \end{equation}
    The associated generator acting on the qubit picture is $\Omega K = 2i Y \otimes \ketbra{m}{m}$, and thus
    \begin{align}
        e^{t\Omega K} &= \sum_{s=0}^{\infty} \frac{(t \Omega K)^s}{s!} = \sum_{s=0}^{\infty} \frac{(2it\, Y\otimes \ketbra{m}{m})^s}{s!} = \id_{2M} + (\cos(2t)-1) \id \otimes \ketbra{m}{m} + i\sin(2t) Y \otimes \ketbra{m}{m}\\
        &=\id\otimes\overline{\ketbra{m}{m}}+\id\otimes \ketbra{{m}}{{m}}  + (\cos(2t)-1) \id \otimes \ketbra{m}{m} + i\sin(2t) Y \otimes \ketbra{m}{m}\\
        &=\id\otimes \overline{\ketbra{m}{m}}+ (\cos(2t)) \id \otimes \ketbra{m}{m} + i\sin(2t) Y \otimes \ketbra{m}{m}\\
        &=\id\otimes \overline{\ketbra{m}{m}}+ e^{2itY} \otimes \ketbra{m}{m}\,.
    \end{align}
    Above we have used the fact that $\id_{2M}=\id\otimes\ketbra{m}{m}+\id\otimes \overline{\ketbra{m}{m}}$ where $\overline{\ketbra{m}{m}} \coloneq \id - \ketbra{m}{m}$ is the projector onto the orthogonal complement of $\ket{m}$. 
    Hence this gate acts trivially when the state in the register qubits is $\ket{m'}$ such that $m'\neq m$, while it applies an $R_y(4t)$ rotation on the symplectic qubit otherwise (here we assume the standard definition $R_y(\theta)=e^{-i \theta Y /2}$). Hence, the associated qubit gate is a SELECT-$R_y(4t)$.

    \item \textbf{Beamsplitter:}
    This gate is described by $ \hat{H}=   \opq_m \opp_{m'}- \opq_{m'}\opp_m $, where $m\neq m'$, in bosonic operators. Therefore the real symmetric $K$ matrix can be expressed as
    \begin{equation} \label{eq-ap:K-beamsplit}
        K=2 \left(   \ketbra{m'+M}{m} + \ketbra{m}{m'+M} - \ketbra{m+M}{m'} - \ketbra{m'}{m+M} \right) = 2iY \otimes ( \ketbra{m}{m'} -  \ketbra{m'}{m})\,.
    \end{equation}
    Therefore $\Omega K = 2i \id \otimes (i \ketbra{m}{m'} - i \ketbra{m'}{m})$. Here, instead of directly exponentiating this operator and finding a closed formula (as we did with the phase gate), we will derive a sequence of gates whose combined actions lead to $e^{\Omega K }$.

    We begin by noting that $\Omega K$  acts trivially on the symplectic qubit. Therefore we focus on the action on the register qubits, where it corresponds to a $y$ rotation in the subspace spanned by $\ket{m}$ and $\ket{m'}$. We write the associated classical bitstrings as $m=m_n\cdots m_1$, and analogously for $m'$. We denote the bitstring operation $\overline{x}$ as taking the bit-wise negation of each individual bit. 
    We separate the bitstring indices between those where the bits of $m$ and $m'$ match, and those where they differ. We call the first set $e=\{e_j\}_{1\leq j \leq E}$, such that $m_{e_j} = m'_{e_j}$ $\forall j$, and the second set $d=\{d_j\}_{1\leq j \leq D}$, such that $m_{d_j} = \overline{m'_{d_j}}$ $\forall j$, where $E+D=n$. We can then factorize Eq.~\eqref{eq-ap:K-beamsplit} to obtain
    \begin{equation}
        \Omega K= 2i\id \otimes \left(\prod_{j=1}^E\ketbra{m_{e_j}}{m_{e_j}}_{e_j} \cdot  (i\ketbra{m_d}{\overline{m_d}}-i\ketbra{\overline{m_d}}{m_d})_d\right) \,,
    \end{equation}
    where the notation $\ketbra{m_{e_j}}{m_{e_j}}_{e_j}$ indicates a projector on  $\ket{e_j}$ on qubit $e_j$ and identity on the rest, and $(i\ketbra{m_d}{\overline{m_d}}-i\ketbra{\overline{m_d}}{m_d})_d$ acts non-trivially on the qubits whose indexes belong in $d$. In fact, it is a $y$-rotation in the subspace spanned by ($\ket{\overline{m_d}},\ket{m_d}$). We (arbitrarily) choose to map this rotation to the least-significant qubit where the $m$ and $m'$ bitstrings differ, that is, on qubit $d_1$. To do so we are going to implement a change of basis using controlled multi-NOTs, so that $\ket{m_d}\rightarrow\ket{0}\otimes\ket{0}^{\otimes{D-1}}$ and $\ket{\overline{m_d}}\rightarrow\ket{1}\otimes\ket{0}^{\otimes{D-1}}$. First, we apply $B \coloneq X_{d_1}^{m_{d_1}}$ so that the least-significant qubit matches the previous expression. Then we apply the two following controlled multi-NOT:
    \begin{align}
        C_0 \coloneq \text{CTRL}(\ketbra{0}{0}_{d_1}) \prod_{j=2}^D X_{d_j}^{m_{d_j}}\,,\qquad
        C_1 \coloneq \text{CTRL}(\ketbra{1}{1}_{d_1}) \prod_{j=2}^D X_{d_j}^{\overline{m_{d_j}}}\,,
    \end{align}
    where $\text{CTRL}(\Pi) U_k$ is the gate $U$ applied to qubits from set $k$ controlled on the single-qubit local projector $\Pi$, for example, $\text{CTRL}(\ketbra{1}{1}_{5})X_1$ is an $X$-gate on qubit $1$ controlled by the qubit $5$. Then we can apply the SELECT-$R_y$ gate with a target on the qubit whose index is $d_1$, whose gate generator is as follows.
    \begin{equation}
      SY \coloneq \prod_{j=1}^E\ketbra{m_{e_j}}{m_{e_j}}_{e_j} \prod_{j=2}^D\ketbra{0}{0}_{d_j} Y_{d_1}\,.
    \end{equation} 
      Finally, we apply the Hermitian conjugate of the change of basis $B^{\dagger} C_0^{\dagger}C_1^{\dagger}=B C_0 C_1$ to return to our original basis. Overall the transformation goes as follows:
      \begin{align}
           &e^{2i t \id \otimes i( \ketbra{m}{m'} -  \ketbra{m'}{m})} = e^{2it \id \otimes C_1C_0B SY BC_0C_1} = \id \otimes C_1C_0B e^{2it SY}BC_0C_1=\\
            &\id \otimes (C_1C_0B) \text{SELECT}\left(\prod_{j=1}^E\ketbra{m_{e_j}}{m_{e_j}}_{e_j} \prod_{j=2}^D\ketbra{0}{0}_{d_j}\right) R_y(4t)_{d_1}  (BC_0C_1)\,.
      \end{align}
    where $\text{SELECT}(\Pi) U_k$ is the gate $U$ applied to qubits from set $k$ controlled on the projector $\Pi$, for example, $\text{SELECT}(\ketbra{0}{0}_{3}\ketbra{1}{1}_{5})X_1$ is an $X$-gate on qubit $1$ controlled by the qubit $5$ and 0-controlled by qubit $3$.

    \item  \textbf{Squeezing Gate:}
    This gate is described by the generator $ \hat{H}= \pm (\opp_m \opq_m + \opq_m \opp_m)$. Therefore the real symmetric $K$ matrix can be expressed as
    \begin{equation}\label{eq-ap:lcu-state}
        K=2(\ketbra{m}{m+M} + \ketbra{m+M}{m}) =  2(\ketbra{0}{1} + \ketbra{1}{0})\otimes \ketbra{m}{m} = 2X \otimes \ketbra{m}{m}\,.
    \end{equation}
    The associated qubit generator is $\Omega K =\pm 2 Z \otimes \ketbra{m}{m}$. It is an imaginary time evolution  $e^{\pm 2t Z \otimes \ketbra{m}{m}}= e^{\mp i t' 2 Z\otimes \ketbra{m}{m}}$, where $t'=it$, under the effective Hamiltonian $\Omega K$.     For small $t$ we have $e^{\pm 2t  Z \otimes \ketbra{m}{m}} = \id \pm 2t Z \otimes \ketbra{m}{m} + O(t^2)$. And therefore we want the state $\ket{\z}$ to be transformed (up to normalization) as
    \begin{equation}
        \ket{\z} \rightarrow (1 \pm 2t) \la \hat{q}_m\ra \ket{0}\otimes\ket{m} + (1 \mp 2t) \la \opp_m\ra \ket{1}\otimes\ket{m} + \sum_{m'\neq m} \la\opq_{m'}\ra \ket{0}\otimes\ket{m'} + \la\opp_{m'}\ra \ket{1}\otimes\ket{m'}\,.\label{eq:ap-state-z}
    \end{equation}
    This can be implemented by a heralded protocol as a Linear Combination of Unitaries (LCU) of the form
    \begin{equation} 
        a \id +  
        b (\id \otimes\ketbra{m}{m} - \id\otimes\overline{\ketbra{m}{m}}) + 
        c (Z\otimes\ketbra{m}{m} + \id\otimes\overline{\ketbra{m}{m}}) + 
        d (-Z\otimes\ketbra{m}{m} + \id\otimes\overline{\ketbra{m}{m}})\,,
    \end{equation}
     with $a + b + c + d = 1$ and $a,b,c,d \geq 0$. 
    Applied to the state $\ket{\z}$ the above LCU yields the following state
    \small
    \begin{equation}\label{eq-ap:im-time-state}
        (a+b+c-d) \la \opq_m\ra \ket{0}\ket{m} + (a+b-c+d) \la\opp_m\ra \ket{1}\ket{m} + (a-b+c+d) \sum_{m'\neq m} \la\opq_{m'}\ra \ket{0}\ket{m'} + \la\opp_{m'}\ra \ket{1}\ket{m'}\,.
    \end{equation}
    \normalsize
    We want this state to be proportional to the one in Eq.~\eqref{eq:ap-state-z} by a factor $\gamma<1$. We thus need to solve the linear system of equations 
    \begin{align}
    \begin{cases}
        a + b + c + d &= 1\\
        a+b+c-d &= \gamma(1\pm2t)\\
        a+b-c+d &= \gamma(1\mp 2t)\\
        a-b+c+d &= \gamma
    \end{cases}
    \,,
    \end{align}
    whose solution is
    \begin{equation}
        a = \frac{3\gamma -1}{2}\,, \qquad b=\frac{1-\gamma}{2}\,,\qquad c= \frac{1-\gamma\pm 2\gamma t}{2}\,, \qquad d=\frac{1-\gamma\mp 2\gamma t}{2}\,.
    \end{equation}
    Since $c$ and $d$ need to be larger than 0,
    \begin{equation}
        \gamma \leq \frac{1}{1\mp 2t} \,, \qquad\gamma \leq \frac{1}{1\pm2t}\,.
    \end{equation}
    In order to maximize the probability of success we should choose the maximum $\gamma$ that satisfies these constraints. Depending on the sign of $\hat{H}$ one or the other inequalities above is saturated. Therefore we choose $\gamma=1/(1+2t)$. This yields
    \begin{equation}
        a= \frac{1-t}{1+2t}\,, \qquad b= \frac{t}{1+2t} \,,\qquad c=\frac{t\pm t}{1+2t}\,, \qquad d= \frac{t\mp t}{1+2t}\,.
    \end{equation}
    Let us calculate the probability of success as the norm of the state~\eqref{eq-ap:im-time-state} for small $t$,
    \begin{equation}
        \frac{1}{(1+2t)^2}\left((1\pm2t)^2\la q_m\ra^2 + (1\mp 2t)^2 \la p_m\ra^2 + (1-\la q_m\ra^2 - \la p_m \ra^2) \right)\approx \frac{1\pm4t\la q_m\ra^2 \mp 4t \la p_m\ra^2}{1+4t}\,.
    \end{equation}
    When $\hat{H}=\opp_m \opq_m + \opq_m \opp_m$, we can see that the best case is $\la q_m\ra=1$ and $\la p_m\ra=0$, which yields a probability of failure of $0$, and the worst case is $\la p_m\ra=1$, which yields a probability of failure of $\sim8t$.

    For a squeezing gate, we are given a bitstring description of the mode to which it applies, and the squeezing parameter $t$. This allows us to compute $(a,b,c,d)$ as per the above equations. In practice, we add two ancillary qubits initialized to zero. We design a unitary $U$ such that $\ket{00} \xrightarrow{U} \sqrt{a}\ket{00} + \sqrt{b}\ket{01} + \sqrt{c}\ket{10} + \sqrt{d} \ket{11}$, and apply it to the ancillary register. Then we apply the SELECT-unitaries as per the LCU, derived above. As one unitary is the identity, we do not need to apply this gate, and as either $c$ or $d$ is zero we do not need to apply the corresponding gate either. Each SELECT gate also has controls on the register to select the mode it is being applied to). If the gate generator $\hat{H}$ has a positive sign, then $d=0$, and we apply successively:
    \begin{align}
        B &\coloneq \text{SELECT}(\ketbra{01}{01}_a) (e^{i\pi})_0 \text{SELECT}(\ketbra{m}{m})_r\\
        C &\coloneq \text{SELECT}(\ketbra{10}{10}_a) Z_0 \text{SELECT}(\ketbra{m}{m})_r\,.
    \end{align}
    We denote $O_{a/0/r}$ operations applied to the ancillary/ symplectic/ register qubit(s). 
    We then apply the hermitian conjugate of the state preparation $U^{\dagger}$, and post-select those states where the ancillary qubits are measured to be the $\ket{00}$ state. Overall the transformation is     $\ketbra{00}{00}_a U_a^{\dagger} B C U_a \ketbra{00}{00}_a$.

\item \textbf{Displacement gate:} Displacement gates cannot be implemented as qubit gates on a single copy of the input states. Indeed displacement implies adding a number to an amplitude, whereas unitaries acting on a single copy of a state can only multiply amplitudes. While access to multiple copies could in principle be used to implement non-linear transformations~\cite{holmes2021nonlinear}, we do not consider this setting here.

\end{itemize}

\subsection{Global bit-structured gates}
\label{app:polytuni}
In this section, we present a list of global bosonic gates that can be easily translated to qubit gates. In particular, the local interferometric gates map to global qubit gates composed of multi-qubit controlled operations. To mitigate this issue, we can combine local GB gates into global ones, such that their qubit counterparts require fewer multi-qubit controls. We will henceforth refer to the GB gates which effectively translate into local qubit operations as \textit{global bit-structured Gaussian gates}.  
\begin{itemize}
    \item \textbf{Phase gate:} It is defined with a binary condition describing which modes the same local gate is applied to. It is given as pairs of indices and binary values. For example, $((1,1),(3,0))$ translates into the binary condition $m_1\overline{m_3}=1$, which means that the least significant bit should be $1$ and the third bit should be $0$. For $2^3$ modes this implies that the rotation gates apply to modes $001=1$ and $011=3$. In the corresponding qubit gate, this bitstring condition directly translates into a SELECT on the register. We denote these gates as $P(((k_j,b_j))_j,t)$. Using the same notations as for the local gates, and using $0$ as the index for the symplectic qubit, for the example $m_1\overline{m_3}=1$ we find 
    \begin{equation}
        P(((1,1)(3,0)),t) \rightarrow \text{SELECT}\left(\ketbra{1}{1}_1 \ketbra{0}{0}_3\right) R_y(4t)_0\,.
    \end{equation}
     The shorter the bitstring condition of the phase gate is, the more local the operation is in qubits (fewer controls in the SELECT) and the less local it is in the interferometer (more modes are acted upon non-trivially). In the case no bit condition is given, the same rotation gate is applied to all modes and therefore it is simply an $R_y$ on the symplectic qubit.
     \item \textbf{Global Beamsplitter:} It is also defined with a binary condition describing which modes the same local gate is applied to. But as a beamsplitter is a two-mode gate, an additional index $l$ is given to determine how the modes are paired. The $l$-th bit cannot be part of the bitstring condition. Each mode whose index satisfies the bitstring condition is paired with the one whose index has all bits in common but the $l$-th one. 

    For example, the global bit-structured beamsplitter on $2^3$ modes described by $((3,0))$, $l=1$ is applied to the second half of the modes ($\overline{m_3}=1$), pairing even modes $0m_20$ with odd modes $0m_21$. Therefore it pairs modes $(000,001)$ and $(010,011)$. We denote these gates as $BS(((k_j,b_j))_j,l,t)$. The example gate may then be written as
    \begin{equation}
        BS(((3,0)),1,t) \rightarrow \text{CTRL}\left(\ketbra{0}{0}_3 \right) R_y(4t)_1\,.
    \end{equation}
    Note that this particular example is of interest because it corresponds to the 0-controlled-$R_y$ gate that is used extensively in the BQP-completeness proof.
    \item \textbf{Squeezing gate:} The modes to which the squeezing applies are also described as a bitstring condition. It is the same as for the rotation gates but with the squeezing apparatus on the ancillary register instead of the $R_y$ on the symplectic qubit. We denote them as $S(((k_j,b_j))_j,r)$.
 \end{itemize}   

Notice that when the binary condition applies to all bits, then we retrieve one local gate from the previous section.

\section{BQP-completeness}

Here we provide proof that Problem 1 is BQP-complete. Let us first recall \Cref{def-ap:interferometer} and \Cref{pbm-ap:interferometer}. 

\begin{definition}[Bit-structured interferometer] A bit-structured interferometer acting on $2^n$ nodes consists of $L$ global beamsplitters, such that each global beamsplitter acts on $2^{n-1}$ modes. A global beamsplitter is specified by two natural numbers, $k\neq l$, between 1 and $n$. The global beamsplitter then acts on all the modes with indices $\{m\}$ such that their $k$-th bit is equal to 0, by applying local beamsplitters between modes with indices $m,m'$ that only differ in their $l$-th bit.
\label{def-ap:interferometer}
\end{definition}

\begin{problem}
\label{pbm-ap:interferometer} 
    Consider a bit-structured interferometer (see \Cref{def:interferometer}) acting on $2^n$ modes with $L\in\OC(\poly(n))$, 
    and an input state such that the first mode is displaced in position by a real constant $x$ while the state of the remaining modes is the vacuum. Then, decide whether the expectation value of the position on the first mode at the output of the interferometer is
    \begin{align}
        1.\,\langle\opq_1\rangle > \frac{2}{3} x\,,\qquad {\rm or} \qquad 2.\, \langle\opq_1\rangle < \frac{1}{3} x\,,\nonumber
    \end{align}
    given the promise that either one or the other is true. 
\end{problem}

We prove our main result by showing that \Cref{pbm-ap:interferometer} reduces to a BQP-complete problem and vice-versa.
\begin{theorem}
     \Cref{pbm-ap:interferometer} is BQP-complete.
\end{theorem}
\begin{proof}
    
We recall the following problem, known to be BQP-complete.
\begin{problem}
    \label{pbm-ap:bqp-complete}
    Given a uniform family of quantum circuits on $n$ qubits with $J\in\OC(\poly(n))$ local gates $\{U_j\}_{1\leq j \leq J}$ taken from a universal gate set $\mathcal{S}$, which are applied to the state $\ket{0}^{\otimes n}$ to produce $\ket{\psi}=\prod_j U_j \ket{0}^{\otimes n}$, decide whether 
    \begin{align}\nonumber
        1. \,\ket{\psi} \text{ has an  overlap larger than 2/3 with } \ket{0}^{\otimes n} \qquad {\rm or} \qquad 
        2.\, \ket{\psi} \text{ has an overlap smaller than 1/3 with } \ket{0}^{\otimes n}\,,
    \end{align}
    given the promise that either one or the other is true.
\end{problem}

    \subsection*{Inclusion in BQP}
    First, we prove that \Cref{pbm-ap:interferometer} is in BQP by showing that \Cref{pbm-ap:interferometer} can be efficiently reduced to \Cref{pbm-ap:bqp-complete}. We are given access to an algorithm to solve \Cref{pbm-ap:bqp-complete} and a bit-structured interferometer composed of polynomially many layers of global beamsplitters. The initial state in \Cref{pbm-ap:interferometer} is a tensor product of coherent states such that $\la \opq_m \ra= x \delta_{m=1}$ and $\la \opp_m \ra = 0$. We consider a quantum circuit over $n+1$ qubits, composed of the symplectic qubit and the $n$ register qubits, as explained in the main text. The initial coherent state then corresponds to an input state $(1,0,\dots,0)^T = \ket{0}^{\otimes n+1}$ in the qubit picture. 

    We recall that a uniform family of quantum circuits is a set of circuits $\{C_n\}$ such that a classical Turing machine can produce a description of $C_n$ on input $n$ in time polynomial in $n$. In our case, the classical description of the beamsplitter gates is a pair of natural numbers smaller or equal than $n$ for a problem of size $n$. As such, this description can be efficiently translated to a circuit description using the dictionary we provided in \Cref{app:gates}, as we know that a single layer of global beamsplitters can be mapped to a 0-controlled-$R_y$ gate on the register (see \Cref{def-ap:interferometer}). Therefore, we can construct a uniform family of quantum circuits implementing the action in phase space of bit-structured interferometers over $2^n$ modes. 
    
    We use access to the solver of \Cref{pbm-ap:bqp-complete} with the polynomial-size sequence of 0-controlled-$R_y$ gates corresponding to the global beamsplitters to determine whether the output state has an overlap with $\ket{0}^{\otimes n+1}$ that is $>2/3$ or $<1/3$. This directly answers the question of whether the final coherent state has a position expectation value for the first mode $>2x/3$ or $<x/3$. We have therefore proved that \Cref{pbm-ap:interferometer} reduces to \Cref{pbm-ap:bqp-complete}, implying that it is in BQP. 
    
    As a side note, in \Cref{app:gates} we show that a broader class of particle-preserving gates can be simulated efficiently by a quantum computer. Indeed we have mapped each local and bit-structured global interferometric gate to a constant number of multi-controlled qubit gates. Each of the multi-controlled gates can be decomposed into $\OC(n)$ two-qubit gates. Therefore any interferometer made of a polynomial number of local or bit-structured global gates over $2^n$ modes can be simulated efficiently by a polynomial-depth circuit acting on $(n+1)$ qubits.

    \subsection*{BQP hardness}
    Second, we prove that \Cref{pbm-ap:interferometer} is BQP-Hard. To do so, we show that \Cref{pbm-ap:bqp-complete} efficiently reduces to \Cref{pbm-ap:interferometer}. As for the inclusion proof, the reduction is based on the fact that the beamsplitters composing a structured interferometer as in \Cref{def-ap:interferometer} translate to 0-controlled-$R_y$ rotations between all pairs of register qubits (as proven in \Cref{app:gates}). The key point for the hardness is that  0-controlled-$R_y$ gates constitute a universal gate set for quantum computation, as stated in the following Lemma (which is a restatement of a result in~\cite{rudolph20022}).
    \begin{lemma}
        The set of 0-controlled-$R_y(\theta)$ rotation gates with control qubit $k$ and target qubit $l$, with $1\leq k\neq l\leq n+2$ and $\theta \in [0,2\pi]$, applied to the initial state $\ket{0}^{\otimes n +2}$, is universal for quantum computation on $n$ qubits.
    \end{lemma}
    \begin{proof}

        Let us suppose that we have an $n$-qubit quantum state 
        \begin{equation} \label{eq-ap:complex-state}
            \ket{\psi} = \sum_{r=0}^{2^n-1} a_r e^{i\theta_r} \ket{r}\,,
        \end{equation}
        where the $a_r$ and $\theta_r$ are real numbers,
        together with the following universal gate set,
        \begin{equation} \label{eq-ap:universal-set}
            R_z(\tau) = \begin{pmatrix}
                e^{i\tau} & 0 \\ 
                0 & 1
            \end{pmatrix} \,, \qquad R_y(\tau)  =\begin{pmatrix}
                \cos (\tau/2) & -\sin (\tau/2) \\ \sin (\tau/2) & \cos (\tau/2)
            \end{pmatrix}\,, \qquad F\left(\frac{\pi}{2}\right)=\begin{pmatrix} \  0 & -1 & 0 & 0\\  1 & 0 & 0 & 0 \\  0 & 0 & 1 & 0 \\  0 &0& 0 & 1  \               
            \end{pmatrix} \,,
        \end{equation}
        where our $R_z(\tau)$ is equivalent to the standard one, as they only differ by a global phase and a relabeling $\tau\rightarrow-\tau$. Moreover, our $F(\pi/2)$ gate can be readily mapped to that in the universal set from Ref.~\cite{rudolph20022}, by using single-qubit $X$ gates. Since $R_z(\tau)$ and $R_y(\tau)$ can generate any single-qubit gate, both gate sets are equivalent and hence universal.
        
        Let us furthermore suppose that we have an $(n+1)$-qubit quantum state with real amplitudes,
        \begin{equation} \label{eq-ap:real-state}
            \ket{\phi} = \sum_{r=0}^{2^n-1} a_r \cos \theta_r \ket{r} \ket{0} + a_r \sin \theta_r \ket{r}\ket{1}\,.
        \end{equation}
        Clearly, the states in Eqs.~\eqref{eq-ap:complex-state} and~\eqref{eq-ap:real-state} contain the same information. We refer to the extra qubit in $\ket{\phi}$ as the ancilla.
        The action of the universal gates~\eqref{eq-ap:universal-set} on $\ket{\psi}$, such that $\ket{\psi}\rightarrow\ket{\psi'}$, will then induce an action $\ket{\phi}\rightarrow\ket{\phi'}$. We need to show that this induced action can be efficiently implemented using controlled-$R_y$ rotations that act non-trivially when the control qubit is in the $\ket{0}$ state.

        We begin with the $R_z$ gate, whose action on a single-qubit state is
        \begin{equation}
            R_z (r_0 e^{i\theta_0} \ket{0} + r_1 e^{i\theta_1}\ket{1}) =  r_0 e^{i(\theta_0+\tau)} \ket{0} + r_1 e^{i\theta_1}\ket{1}\,.
        \end{equation}
        The induced evolution is
        \begin{align}
            r_0 \cos \theta_0 \ket{0}\ket{0} + r_0 \sin\theta_0 \ket{0} \ket{1} &+ r_1\cos\theta_1\ket{1}\ket{0} +  r_1\sin\theta_1\ket{1}\ket{1} \nonumber \\ &\downarrow \nonumber \\  r_0 \cos (\theta_0+\tau) \ket{0}\ket{0} + r_0 \sin (\theta_0+\tau) \ket{0} \ket{1} &+ r_1\cos\theta_1 \ket{1}\ket{0} +  r_1\sin\theta_1\ket{1}\ket{1}\,.
        \end{align}
        This can be achieved by performing a 0-controlled-$R_y(\tau)$ gate where the control is the first qubit (i.e., the qubit on which $R_z(\tau)$ would act) and the target is the ancilla.

        Next, let us look at the action of $R_y(\tau)$. To implement this gate, we simply need an additional auxiliary qubit in the $\ket{0}$ state and to apply a 0-controlled-$R_y$ gate conditioned on this extra qubit.

        Finally, we have the $F\left(\frac{\pi}{2}\right)$ gate, whose action on a two-qubit state is given by
        \begin{align}
            r_0 e^{i\theta_0} \ket{00} + r_1 e^{i\theta_1} \ket{01} &+ r_2 e^{i\theta_2} \ket{10} +r_3 e^{i\theta_3} \ket{11}  \nonumber \\ &\downarrow \nonumber \\ r_1 e^{i\theta_1} \ket{00} - r_0 e^{i\theta_0} \ket{01} &+ r_2 e^{i\theta_2} \ket{10} + r_3 e^{i\theta_3} \ket{11}\,.
        \end{align}
        The corresponding induced action is 
        \begin{align}
            r_0 \cos\theta_0 \ket{000} + r_0 \sin\theta_0 \ket{001}  &+ r_1 \cos \theta_1 \ket{010}  + r_1 \sin \theta_1 \ket{011}+ \nonumber \\   r_2 \cos \theta_2 \ket{100} + r_2 \sin \theta_2 \ket{101}  &+ r_3 \cos\theta_3 \ket{110} + r_3 \sin\theta_3 \ket{111}
            \nonumber \\ &\downarrow\nonumber \\   r_1 \cos\theta_1 \ket{000} + r_1 \sin\theta_1 \ket{001}  &- r_0 \cos \theta_0 \ket{010}  - r_0 \sin \theta_0 \ket{011}+ \nonumber \\   r_2 \cos \theta_2 \ket{100} + r_2 \sin \theta_2 \ket{101} & + r_3 \cos\theta_3 \ket{110} + r_3 \sin\theta_3 \ket{111} \,.
        \end{align}
        The previous can be achieved by simply applying a 0-controlled-$R_y(\pi)$ in the first two qubits (i.e. the ancilla is not necessary).
        Therefore, we conclude that the set of $0$ is universal for quantum computation, given the initial state $\ket{0}^{\otimes n+2}$.
        \end{proof}

    We are given a circuit composed of a polynomial number of  0-controlled-$R_y$ gates. We are also given access to a solver for \Cref{pbm-ap:interferometer}. We add a qubit on top of the given circuit which is acted trivially upon, and consider it as the symplectic qubit. We query the given solver with the sequence of global bit-structured beamsplitters corresponding to the sequence of 0-controlled-$R_y$ gates as input. Similarly to the inclusion proof, because the input state is the all-zero state, the reduction directly follows. We conclude that \Cref{pbm-ap:interferometer} is BQP-complete.
\end{proof}

\section{From unitary quantum circuits to interferometers}
\subsection{Separating real and imaginary parts of the amplitudes of a quantum state}

\label{app:unit2symplectic}
Consider a unitary quantum circuit on $n$ qubits. Such a circuit is applied to a complex state on $n$ qubits of the following form
\begin{equation}
    \ket{\psi} = \sum_{r=0}^{2^n-1} (a_r + b_r i) \ket{r}\,.
\end{equation}
Adding one qubit (as the left-most in the tensor product) which we call the symplectic qubit for reasons that will become clear later, we can define the real-valued state over $n+1$ qubits. 
\begin{equation}
    \ket{\phi} = \sum_{r=0}^{2^n-1} a_r\ket{0}\ket{r} + b_r\ket{1}\ket{r}\,.
\end{equation}

First, we recall that a set of universal one-qubit gates, together with any entangling gate forms a universal set. We can use $Rz$ and $Ry$ gates to generate any $Rx$ gate we wish, and thus $Rz$ and $Ry$ form a universal gate set for unitaries. Therefore together with $CRy$, they form a universal gate set for unitaries on $n$ qubits. This yields the following lemma.
\begin{lemma}
    The set of gates $\{Rz,Ry,CRy\}$ is universal.
\end{lemma}

 We are going to prove that this universal set of gates $\{Ry, Rz, CRy\}$ for unitary circuits can be translated to a specific set of orthogonal gates on $n+1$ qubits. It is easy to see that for any real gate, such as $Ry$, the gate is simply applied to the register.
\begin{equation}
    Ry(\tau) = \begin{pmatrix}
    \cos(\tau/2) & -\sin(\tau/2)\\
    \sin(\tau/2) & \cos(\tau/2)
    \end{pmatrix}
    \rightarrow
    \id \otimes Ry(\tau/2) = \begin{pmatrix}
    \cos(\tau/2) & -\sin(\tau/2)& 0 & 0\\
    \sin(\tau/2) & \cos(\tau/2) &0&0\\
    0&0&\cos(\tau/2) & -\sin(\tau/2)\\
    0&0&\sin(\tau/2) & \cos(\tau/2)
    \end{pmatrix}
    =\exp(-i\tau\id\otimes Y/2).
\end{equation}
This is also true for controlled-$Ry$, which yields the same controlled-$Ry$ on the register
\begin{equation}
    C_k{Ry}_l(\tau/2) 
    \rightarrow
    \id \otimes C_k{Ry}_l(\tau/2) 
    =\exp(-i\tau\id \otimes \ketbra{1}{1}_k Y_l/2).
\end{equation}
For complex gates, the symplectic qubit is involved in the corresponding orthogonal gate. We show a derivation for $Rz$ below,
\small
\begin{equation}
    Rz(\tau/2) \ket{a+ic,b+id}= 
    \begin{pmatrix}
    1 & 0\\
    0 & \cos(\tau/2) + i \sin(\tau/2)
    \end{pmatrix}
    \begin{pmatrix}
    a + ic\\
    b + i d
    \end{pmatrix}
    =
    \begin{pmatrix}
    a + ic\\
    (\cos(\tau/2) b + \sin(\tau/2) d) + i (-\cos(\tau/2) d + \sin(\tau/2) b)
    \end{pmatrix}\,.
\end{equation}
\normalsize
Finally we conclude that $Rz(\tau/2)\rightarrow C_1RY_0(\tau/2) = \exp(-i\tau Y\otimes \ketbra{k}{k}/2)$
\begin{equation}
    C_1RY_0(\tau/2) \ket{a,b,c,d} \rightarrow
    \begin{pmatrix}
    1 & 0 & 0 & 0\\
    0 & \cos(\tau/2) & 0 & -\sin(\tau/2)\\
    0&0& 1 & 0\\
    0&\sin(\tau/2)&0 & \cos(\tau/2)
    \end{pmatrix}
    \begin{pmatrix}
    a \\
    b \\
    c \\
    d
    \end{pmatrix}.
\end{equation}
We take note that for all of the gates in a universal set, we have derived an orthogonal gate that belongs to the unitary symplectic algebra as in \Cref{app:symplectic}, with the generators of $\id \otimes Ry \in \{i \id \otimes P_a\}$, those of  $\id \otimes CRy \in\{ i \id \otimes P_a\}$ and those of $Rz \rightarrow C_k{Ry}_0 \in \{iY \otimes P_s\}$.

\subsection{Unitary gates to global bit-structured interferometric gates} 

We now show that each of the gates derived in the previous section maps in turn to a global bit-structured interferometric gate, as follows. 
\begin{itemize}
    \item An $Rz$ gate on qubit $k$ gate is mapped to a $C_kRy_0$ gate in the symplectic picture. Its corresponding GB gate is the phase gate on half of the modes, whose $k$-th bit of the index is 1. Using notation from \Cref{app:gates} it is $\text{P}((k,1))$.
    \item An $Ry$ gate on qubit $k$ gate is mapped to a $Ry_k$ gate in the symplectic picture. Its corresponding GB gate is the beamsplitter where all modes are paired such that their index differs only by their $k$-th qubit. Using notation from \Cref{app:gates} it is $\text{BS}(\varnothing,k)$.
    \item A $CRy$ gate controlled on qubit $l$ and target on qubit $k$ gate is mapped to a $C_lRy_k$ gate in the symplectic picture. Its corresponding GB gate is the beamsplitter where all the modes whose $l$-th bit of the index is $1$ are paired such that their indices only differ by their $k$-th qubit. Using notation from \Cref{app:gates} it is $\text{BS}((l,1),k)$.
\end{itemize}

We summarize the previous equivalences in the below table.
\begin{table}[h]
    \centering
\begin{tabular}{ |c |c |c |}
\hline
Unitary gate & Symplectic gate & GB gate\\
\hline
 $Rz_k$ & $C_kRy_0$ & $\text{P}((k,1))$ \\
 \hline
 $Ry_k$ & $\id \otimes Ry_k$ & $\text{BS}(\varnothing,k)$ \\ 
 \hline
 $C_lRy_k$ & $\id \otimes C_lRy_k$ & $\text{BS}((l,1),k)$ \\ 
 \hline
\end{tabular}
\caption{Mapping between unitary, symplectic and bosonic gates.}
    \label{tab:my_label}
\end{table}

\subsection{From unitary circuits to interferometers}
In this subsection, we show how a unitary qubit computation can be mapped to the evolution by an interferometer of the first moments of expectation values of quadrature operators of coherent states over exponentially many modes. 

\textbf{Gates.} Based on the derivations above, given a unitary on $n$ qubits composed of $L$ gates from the universal set $\{Rz,Ry,CRy\}$, we can map it to an equivalent interferometer on $2^n$ modes composed of exactly $L$ global bit-structured interferometric gates (global bit-structured beamsplitters and global bit-structured phase gates). In that picture, the $k$-th mode tracks the amplitude of the qubit state on the computational-basis state $\ket{m}$ with the position as the real part and the momentum as the imaginary part. 

\textbf{State preparation.} To prepare a sparse qubit state, we start from the vacuum, and each non-zero entry $a+ib$ is position displaced by $a$ and momentum displaced by $b$ for the corresponding mode. Note we have a degree of freedom to upload a state $\ket{\psi}$ onto our modes up to a multiplicative coefficient. For example the state $((1+i) \ket{0} - \sqrt{2} \ket{1})/2$ can be prepared as $q_0=1, p_0=1, q_1=\sqrt{2}$ but also as $q_0=10, p_0=10, q_1=10\sqrt{2}$. We use the expectation value of the sum of the number operator for each mode $\langle\hat{n}_m\rangle= \frac{1}{2}\langle \sum_m \opq_m^2 + \opp_m^2 \rangle$ to characterize this degree of freedom when encoding qubit states into bosonic states. This expectation value also corresponds to the number of photons $P=\sum_m \langle\hat{n}_m\rangle$ in the circuit.  We recall that coherent states are eigenstates of the anihilation operator $\opa\ket{\alpha} = \alpha \ket{\alpha}$, therefore $\langle \hat{n} \rangle = \langle \opb\opa \rangle = \lvert \alpha \rvert^2 = \frac{1}{2} \left( \la \opq\ra^2 + \la \opp \ra^2 \right)$.

\textbf{Measurements.} We consider the photon counting measurement and show that it corresponds to sampling bitstrings from the qubit circuit. The probability of detecting $p$ photons on the $m$-th mode is a Poissonian distribution $\Pr(p) = \exp(-e_m)e_m^{p}/p! $ with an average equal to the energy of the mode $e_m = \langle \hat{n}_m \rangle$. We recall that the Poissonian distribution expresses the probability of a given number of events occurring when these events occur independently at a known constant mean rate, which is in that case $e_m$. Therefore considering that a number of photons $P = \sum_m \la \hat{n}_m \ra$ has been injected at the beginning of the interferometer we should get $P$ photons at the output, distributed according to a compound Poissonian distribution where each mode has a rate of occurrence $e_m$. Effectively we are getting $P$ bitstring samples according to the distribution $[e_0, e_1, \cdots, e_m]$. Recalling that $e_m = \frac{1}{2} \left( \la \opq\ra^2 + \la \opp \ra^2 \right)$, which is effectively proportional to the probability of sampling the bitstring $m$ at the output of the qubit circuit. The more energy injected at the beginning of the circuit the more samples we get.

Now we consider homodyne detection which measures in the basis of $\opp$, $\opq$ or any combination of the two $\hat{x} = \cos{\theta}\opq + \sin{\theta} \opp$. This yields a Gaussian distribution centered around $\langle \hat{x} \rangle$, and for coherent states, with variance $1/2$. Effectively this is equivalent to doing a Hadamard test to access either the real part or the imaginary part of the amplitude of a state on the computational basis, which is affected by shot noise. Increasing the energy at the input of the interferometers increases the precision with which we can measure the real and imaginary parts of the amplitude on $\ket{m}$.

\end{document}